\documentclass[a4paper,12pt]{article}
\usepackage[english]{babel}
\usepackage{amsmath,amssymb,amsfonts,amsthm, textcomp}

\def\softd{{\leavevmode\setbox1=\hbox{d}%
     \hbox to 1.05\wd1{d\kern-0.4ex{\char039}\hss}}}

\begin{document}
\topmargin -70pt %
\textheight = 720pt %
\makeatletter
\renewcommand{\section}{\@startsection{section}{1}{0pt}%
{3.5ex plus 1ex minus .2ex}{2.3ex minus .2ex}%
{\normalfont\normalsize\bfseries}}
\renewcommand{\subsection}{\@startsection{subsection}{2}{0.7cm}%
{1.5ex plus 0.2ex minus .2ex}{-1ex}%
{\normalfont\normalsize\bfseries}} %
\makeatother

\setcounter{secnumdepth}{2} %
\numberwithin{equation}{section}
\newtheorem{theorem}{\indent Theorem}
\newtheorem{proposition}{\indent Proposition}[section]
\newtheorem{lemma}{\indent Lemma}[section]
\newtheorem{corollary}{\indent Corollary}[section]
\newtheorem{observation}{\indent Observation}[section]
\renewcommand{\proofname}{\indent \textup{\textbf{Proof.}}}

\newcommand{\Card}{\mathrm{Card}}
\newcommand{\Anc}{\mathrm{Anc}}
\newcommand{\Norm}{\mathrm{Norm}}
\newcommand{\EqAOF}{\mathrm{EqAOF}}

\begin{center}
{\LARGE Almost overlap-free words \\and the word problem \\
for the free Burnside semigroup \\satisfying $x^2=x^3$}\\[4pt]
{Andrew N. Plyushchenko, Arseny M. Shur}
\end{center}

\begin{abstract}
In this paper we investigate the word problem of the free Burnside semigroup satisfying $x^2=x^3$ %
and having two generators. Elements of this semigroup are classes of equivalent words. %
A natural way to solve the word problem is to select a unique ``canonical'' representative %
for each equivalence class. We prove that  overlap-free words and so-called %
almost overlap-free words (this notion is some generalization of the notion of overlap-free words) %
can serve as canonical representatives for corresponding equivalence classes. We %
show that such a word in a given class, if any, can be efficiently found. %
As a result, we construct a linear-time algorithm that partially solves %
the word problem for the semigroup under consideration.
\end{abstract}

\section*{Introduction}
\label{Intr}
The free Burnside semigroups satisfying $x^n=x^{n+m}$ are defined by the identities, which impose %
the equivalence of $n$-th and $(n{+}m)$-th powers of words in any context. Thus %
elements of these semigroups are classes of equivalent words. The structure %
of free Burnside semigroups is far from being completely described. However, %
a considerable progress was achieved in the 1990's. Ka\softd{}ourek and Pol\'{a}k~\cite{Kadourek}, %
de~Luca and Varricchio~\cite{Luca}, McCammond~\cite{McCammond}, Guba~\cite{Guba1, Guba2}, %
and do~Lago~\cite{Lago1,Lago2} produced a series of papers that led to the discovery %
of many structural properties of free Burnside semigroups. The reader is referred %
to the survey~\cite{LaSi} for the history and the formulations of remarkable results.

The word problem is probably the most important and challenging combinatorial problem related %
to free Burnside semigroups. It is formulated as~follows: \textit{given words $U$ and $V$, %
decide whether or not $U$ and $V$ are equivalent in a given semigroup}. In a series %
of papers~\cite{Guba1,Guba2,Lago1,Luca,McCammond} the word problem was solved for all %
free Burnside semigroups satisfying $x^n=x^{n+m}$ for~$n\geq3$ and~$m\geq1$. In this case, %
all equivalence classes are regular languages, and the deciding algorithm for the word problem constructs an
automaton recognizing the class of the word $U$ and tries to accept $V$ by this automaton
(see~\cite{Guba1,Guba2}). Due to Green and Rees \cite{Green} and later Ka\softd{}ourek and Pol\'{a}k
\cite{Kadourek}, the word problem for the case $n=1$ was solved modulo periodic groups (i. e., %
reduced to the word problem for the groups satisfying $x^{m}=1$). For the case $n=2$, this problem remains open
(and some equivalence classes are not regular languages, see~\cite{Lago2}). Note that the word problem for the
particular case $n=2$ and $m=1$ was explicitly formulated by Brzozowski \cite{Brz}. This case was considered to be
the hardest one to analyze (see~\cite{Lago1}). In what follows, we consider the free Burnside semigroup satisfying
$x^2=x^3$ and having two generators.

A natural way to solve the word problem is to select a unique ``canonical'' representative for each equivalence
class. Thus every word is equivalent to exactly one ``canonical'' word. %
If the latter can be efficiently found, then the word problem is decidable. It is clear that the choice of the
canonical representatives is not an easy task. For example, we cannot take just a cube-free word as a
representative, since there exist equivalent cube-free words. And if we take the shortest word in the class as a
representative, it may be very hard to determine such a word.

As was proved in \cite{MyWork}, overlap-free words can serve as %
canonical representatives for corresponding equivalence classes. In this paper we %
generalize this result for ``almost'' overlap-free words and show that %
such a word in a given class, if any, can be efficiently found. %
Thus we construct an efficient (in fact, linear-time) algorithm that partially solves %
the word problem for the semigroup under consideration.

To give precise formulations of the main results, we say a few words about definitions %
and notation.

Let $\Sigma=\{a,b\}$. As usual, we write~$\Sigma^*$ %
for the monoid of all words over~$\Sigma$ (including the empty %
word~$\lambda$) and~$\Sigma^+$ for the semigroup of all non-empty words %
over~$\Sigma$. For a word $W$, its length is denoted by~$|W|$ and %
its~$i$-th letter  is denoted by~$W[i]$; thus, $W=W[1]\dotso W[|W|]$. %
In the sequel, we write $W[i\dotso j]$ instead of $W[i]\dotso W[j]$. %
\textit{Factors}, \textit{prefixes}, \textit{suffixes}, and \textit{powers} %
of a word are defined in the usual way. Recall that the \textit{Kleene star} $W^*$ of $W$ is %
the union of all nonnegative powers of the word $W$.

A factor, prefix, or suffix of a word $W$ is called  %
\textit{proper} if it is not equal to $W$.  A factor of a word $W$ %
is called \textit{internal} if it is neither a prefix nor a suffix of $W$.  %
We write  $U\leq V$ ($U<V$, $U\ll V$) if a word $U$ is a factor (resp., %
proper factor, internal factor) of a word $V$. %

A word $U$ is called \textit{overlap-free\/} if it contains %
no factor of the form $XYXYX$ for any $X\in\Sigma^+,Y\in\Sigma^*$. %
If~$U$ contains no \textbf{proper} factor of the form %
above, then we call it \textit{almost overlap-free}. %
Finally, we call a word \textit{cube-free} if it contains no factor of the %
form $XXX$ for any $X\in\Sigma^+$.

Two words $U$ and $V$ are \textit{neighbours} if and only if one %
of them can be obtained from another one by replacing some factor of the form $Y^2$ by the factor $Y^3$. %
Thus we have the \textit{neighbourhood relation} \[\pi = \{(XY^2Z, XY^3Z), (XY^3Z,XY^2Z)\mid X,Z\in\Sigma^*, Y\in\Sigma^+\}.\]  %
The \textit{free Burnside semigroup satisfying~$x^2=x^3$ generated %
by~$\Sigma$} is defined as the quotient semigroup~$\Sigma^+/{\sim}$,
where~$\sim$ is the smallest congruence containing~$\pi$. %
If $U\sim V$, then the words $U$ and $V$ are said %
to be \textit{equivalent}. The congruence class of $U$ is denoted by~$[U]$. \medskip

In these terms, the first result of this paper is formulated as follows.

\begin{theorem}
\label{Th1} %
Except for the classes $[aa]$ and $[bb]$, each class of the congruence~$\sim$ %
contains at most one almost overlap-free word. Each of the two exceptional %
classes contains exactly two almost overlap-free words. %
\end{theorem}

Theorem~\ref{Th1} for overlap-free words was proved in~\cite{MyWork}. %
In this paper the proof from~\cite{MyWork} is simplified and applied to a more general case.\medskip %

The second result of this paper now follows.

\begin{theorem}
\label{Th2} %
If at least one of words $U$ and $V$ is equivalent to an almost overlap-free %
word, then the word problem for the pair $(U,V)$ can be solved in time $O(n)$, where $n=\max\{|U|,|V|\}$.
\end{theorem}

In fact, we will construct Algorithm EqAOF %
(abbr. Equivalent Almost Overlap-Free), which returns the %
almost overlap-free word $V$ that is equivalent to a given input word $U$ %
or reports that no such word $V$ exists. %
It should be mentioned that some %
equivalence classes  contain no almost overlap-free words like %
the class $[ababaa]=(ab)^*ababaa$.

Sketches of the proofs of Theorems~\ref{Th1} and~\ref{Th2} were given in~\cite{Words2007}. %
Here we present a full version of these proofs.

The text is subdivided into six sections. In Sect.~\ref{Prelim} we introduce %
the main tools and techniques. Sect.~\ref{ProofOne} contains the proof of Theorem~\ref{Th1}. %
Last four sections are devoted to the construction and analysis of Algorithm EqAOF.

\section{The main tools and techniques}
\label{Prelim} %
Recall that \textit{Thue-Morse morphism} $\varphi$ %
of $\Sigma^+$ is defined by the rule~\[\varphi(a)=ab,\quad\varphi(b)=ba.\] %
A \textit{$\varphi$-image\/} is any word $U\in\varphi(\Sigma^+)$. %


One can easily check %

\begin{observation}
\label{EvenPhi}
The set $\varphi(\Sigma^+)$ consists exactly of all even-length words such that %
all factors $aa$ or $bb$ start at even positions. %
\end{observation}
\medskip

The main idea of our solution to the word problem for an instance $(U,V)$ is to simplify %
and shorten the words $U$ and $V$. Each simplification is either an equivalent transformation of a word or a
simultaneous transformation of a pair of words, preserving their equivalence/non-equivalence. The main instrument
is the function $\varphi^{-1}$ applied to a pair of words. %
This function reduces the length of both words to one half. All other transformations are needed %
to get a pair $(U',V')$ of $\varphi$-images $U'$ and $V'$. These transformations are
\begin{itemize}
\item[--] complete reduction of a word, which is an equivalent transformation on the class of so-called $\widetilde{AB}$-whole words; %
\item[--] tail reductions applied to a pair of words in order to make the words $\widetilde{AB}$-whole;
\item[--] functions $\xi$ and $\eta$ applied to a pair of completely reduced words in order to turn %
them into $\varphi$-images.
\end{itemize}

In this preliminary section we introduce all these transformations and study their properties.

\subsection{Uniformity and complete reduction.}
\label{UnCompleteRed}
Following~\cite{Shur}, we generalize the notion of $\varphi$-image. %
Namely we say that a word $U$ is \textit{uniform\/} if all its factors %
$aa$ or $bb$ start in $U$ either always at even positions or %
always at odd positions. Otherwise a word is called \textit{non-uniform}. %
We call a word \textit{letter-alternating} if it contains no factor $aa$ or $bb$. %
In what follows, $\widetilde{A}$ (resp., $\widetilde{B}$) abbreviates an arbitrary
letter-alternating word of the form $aba(ba)^*$ (resp., $bab(ab)^*$).  %
All letter-alternating words are obviously uniform. %

Uniformity plays a crucial role in subsequent %
considerations. First, uniform words form a vast majority among %
all almost overlap-free words. Second, an important  connection %
between the uniformity, Thue-Morse morphism, and the congruence~$\sim$
was established in~\cite{Bakirov}: %

\begin{proposition}
\label{MainProp} Suppose $W\sim\varphi(U)$ for some $U\in\Sigma^+$ and %
a uniform word $W$. Then there exists a word $V\in\Sigma^+$ such %
that $W=\varphi(V)$ and $V\sim U$.
\end{proposition}

According to Proposition~\ref{MainProp}, since Thue-Morse morphism preserves the congruence~$\sim$, %
we have $\varphi(V)\sim \varphi(U)$ if and only if $V\sim U$. %
Thus the word problem for two $\varphi$-images %
$U$ and $V$ can be reduced to the word problem for the pair of shorter %
words $\varphi^{-1}(U)$ and $\varphi^{-1}(V)$. So, we are going to replace considered words %
by $\varphi$-images whenever it is possible.
First, we introduce three reduction operations in order to transform any given word %
to a uniform word (see~\cite{Bakirov}).

The first operation is called \textit{$r_1$-reduction\/}. It reduces all factors of the %
form $c^n$ to $c^2$, where $c\in\Sigma$ and~$n>2$. The result of this %
operation applied to a word $U$ is denoted by $r_1(U)$. We say that %
a word $U$ is \textit{$r_1$-reduced\/} if $U=r_1(U)$. %

Obviously, $U\sim r_1(U)$ for any $U\in\Sigma^+$. Moreover, $r_1$-reduction %
preserves the relation $\pi$ as well: $(U,V)\in\pi$ implies $(r_1(U),r_1(V))\in\pi$ or $r_1(U)=r_1(V)$ %
(see~\cite{Bakirov}). Thus if we denote by $\pi_{r_1}$ and $\sim_{r_1}$ %
the restrictions of relations $\pi$ and $\sim$ respectively to %
the set of all $r_1$-reduced words, then we get $\sim_{r_1}=\pi_{r_1}^+$.
Therefore it is sufficient to solve the word problem %
for $r_1$-reduced words only. For this reason, in the sequel we usually consider %
$r_1$-reduced words. %

Let us denote the class of all $r_1$-reduced words that are equivalent %
to $U$ by $[U]_{r_1}$. Obviously, all almost overlap-free words, except for the words $aaa$ and $bbb$, %
are $r_1$-reduced. Hence, Theorem~\ref{Th1} can %
be reformulated as follows: %
\begin{itemize}%
\item[$(*)$] \textit{for any word $U$, the class $[U]_{r_1}$ contains at most one almost overlap-free word.} %
\end{itemize} %
Actually we will prove Theorem~\ref{Th1} in this form.\medskip

\begin{subequations}\label{rAB}%
Now let $U$ be an $r_1$-reduced word. Then $r_A(U)$ is the word obtained from $U$ by performing all possible
reductions of the form
\begin{equation}
a\widetilde{A}a\quad\to\quad aa. \label{rA} %
\end{equation}
The word $r_B(U)$ is defined in a symmetric %
way using reductions of the form
\begin{equation}
b\widetilde{B}b\quad\to\quad bb\ldotp \label{rB} %
\end{equation}
\end{subequations}

Finally, let $r(U)=r_B(r_A(r_1(U)))$ for an arbitrary word $U\in\Sigma^*$. %
The operation $r$ is called  \textit{complete reduction}. We call a word $U$ %
\textit{completely reduced\/} if $r(U)=U$. As it was shown in~\cite{Bakirov}, the word $r(U)$ %
can be obtained from $r_1(U)$ by performing all possible reductions of the form~\eqref{rA} %
and~\eqref{rB} \textbf{in any order}.

It is easy to check that, in contrast to $r_1$-reductions, some complete reductions %
do not preserve the congruence $\sim$. However, %
as we will see later, under certain conditions even the complete reduction preserves $\sim$. \medskip %

It appears that completely reduced words are exactly uniform ones: %

\begin{proposition}
\label{RegRed} The following three conditions are equivalent for an arbitrary word $U$: %

\textup{(1)}\ $U$ is a factor of a $\varphi$-image;

\textup{(2)}\ $U$ is uniform;

\textup{(3)}\ $U$ is completely reduced.
\end{proposition}

The equivalence of (2) and (3) was proved in~\cite{MyWork}, while the equivalence %
of (1) and (2) follows from definitions.

\subsection{Reduction of $\widetilde{AB}$-whole words. Non-reducible tails.}
\label{RedAB}
So, we can reduce any word $U$ to the uniform word $r(U)$.  %
Proposition~\ref{RedSave} below states that words $U$ and $r(U)$ are equivalent %
under certain conditions. This proposition uses the important notions of  %
$\widetilde{A}$-whole and $\widetilde{B}$-whole words (see~\cite{Bakirov}).
An $r_1$-reduced word $W$ is called \textit{$\widetilde{A}$-whole\/} if every %
factor~$X$ of the form $a\widetilde{A}a$ occurs in~$W$ inside the factor~$abXba$. %
The notion of \textit{$\widetilde{B}$-whole word\/} is dual to %
the above one. If a word is $\widetilde{A}$-whole  and $\widetilde{B}$-whole, %
it will be called  $\widetilde{AB}$-whole. Obviously, any completely reduced word is %
$\widetilde{AB}$-whole. %

In the sequel we often use the following proposition proved in~\cite{Bakirov}.

\begin{proposition}
\label{WholeEqv} %
If $r_1$-reduced words $U$ and $V$ are equivalent and %
$U$ is   $\widetilde{A}$-whole ($\widetilde{B}$-whole), %
then $V$ is $\widetilde{A}$-whole (resp., $\widetilde{B}$-whole) as well.
\end{proposition}

Now we ready to formulate Proposition~\ref{RedSave}.

\begin{proposition}
\label{RedSave} %
Let $U$ be an $\widetilde{AB}$-whole word. %
Then $U\sim r(U)$ if and only if $U$ has no prefix of the form $(aba)(aba)^*(ab)^2(ab)^*aa$ and %
no suffix of the form $aa(ba)^*(ba)^2(aba)^*(aba)$ up to negation. %
\end{proposition}

We refer to the prefixes and suffixes mentioned in Proposition~\ref{RedSave} %
as \textit{non-reducible tails} and distinguish four kinds of such tails according to the %
following table:

\centerline{
\begin{tabular}{rcc}
&$A$-tail&$B$-tail\\
left (prefix)&$(aba)(aba)^*(ab)^2(ab)^*aa$&$(bab)(bab)^*(ba)^2(ba)^*bb$\\
right (suffix)&$aa(ba)^*(ba)^2(aba)^*(aba)$&$bb(ab)^*(ab)^2(bab)^*(bab)$
\end{tabular}
}\medskip

In order to prove Proposition~\ref{RedSave}, we need an auxiliary result.


\begin{lemma}
\label{RedTails} %
Let $U$ and $V$ be $r_1$-reduced words, $U\sim V$, and let $U$ have a non-reducible tail. %
Then the word $V$ has a non-reducible tail of the same kind.
\end{lemma}

\begin{proof}
It is sufficient to prove the statement of the lemma  for a pair of neighbours. %
Let $U=XY^kZ$, $V=XY^lZ$, where $X,Z\in\Sigma^*$, $Y\in\Sigma^+$, and  $\{k, l\}=\{2, 3\}$. %
Without loss of generality assume that $U$ has a left $A$-tail, that is, %
$U=(aba)^n(ab)^maaT$, where $n\geq 1, m\geq 2$, and $T\in\Sigma^*$. The words $U$ and $V$ have %
the common prefix $XYY$, therefore the proof is evident if $(aba)^n(ab)^maa$ is a prefix %
of the word $XYY$. So suppose that $XYY$ is a proper %
prefix of the word $(aba)^n(ab)^maa$.  Consider all possible cases.

First suppose that $|X|\geq|(aba)^n|$. In this case, $YY\leq (ab)^ma$. %
Hence $Y$ is a letter-alternating word of even length. Since $Y^k$ is also letter-alternating, %
we have $Y^k\leq (ab)^ma$. Then the word $V$ has the prefix $(aba)^nSa$ for some neighbour %
$S$ of the word $(ab)^ma$. Thus $V$ has a  left $A$-tail, as desired.

Now let $|X|<|(aba)^n|$ and $|XY|>|(aba)^n|$. Then $Y$ contains the factor $aa$ obtained %
from the last letter of $(aba)^n$ and the first letter of $(ab)^ma$.  %
On the other hand, the suffix $Y$ of  $XYY$ is letter-alternating as a factor %
of the word $(ab)^ma$. We get a contradiction. Hence, this case is impossible.

Finally, let $|XY|\leq |(aba)^n|$. Then the word $Y$ contains no factors %
$abab$ and $baba$. Thus, $XY^2\leq (aba)^n(aba)$ whence the word $Y^2$ %
also contains no  factors $abab$ and $baba$, in particular, $Y\not\in\{ab, ba\}$.  %
Since the words $U$ and $V$ are $r_1$-reduced, we get $Y\not\in\{a,b,aa,bb\}$. %
Therefore, $|Y|>2$ and the factors $abab$ and $baba$ can occur %
in $Y^k$ inside the factor $Y^2$ only. %
Hence the word $Y^k$ contains no factors $abab$ and $baba$ as well, and we %
have~$XY^k\leq(aba)^{n+1}$. %
So the word $V$ has the prefix $Sb(ab)^{m-2}aa$ for some neighbour $S$ of %
the word $(aba)^{n+1}$. One can easily check that the prefix $Sb(ab)^{m-2}aa$ is a left $A$-tail.
This completes the proof of the lemma.
\end{proof}


\begin{proof}[\indent \textup{\textbf{Proof of Proposition~\ref{RedSave}}}]
Let a word $U$ be $\widetilde{AB}$-whole  and  have no non-reducible tails.
Recall that reductions of kind~\eqref{rA} and~\eqref{rB} can be applied in any order. %
So, to prove the forward implication it is sufficient to find a sequence of reductions %
from $U$ to $r(U)$ such that every single reduction preserves the relation $\sim$. %
Indeed, by Proposition~\ref{WholeEqv}, %
the word obtained by a single reduction will remain $\widetilde{AB}$-whole and, %
by Lemma~\ref{RedTails}, this word will have no non-reducible tails.

We will reduce $U$ as follows: first perform all possible reductions of the form $aabaa\to aa$ %
and $bbabb\to bb$ (in any order) and then all other reductions (also, in any order).

Suppose that $U=XaabaaY$, and $U'=XaaY$ is obtained from $U$ %
by the reduction $aabaa\to aa$. Since $U$ is $\widetilde{A}$-whole, %
we have $U=X'abaabaabaY'$, where $X = X'ab$ and $Y=baY'$. %
Thus we get $U=X'(aba)^3Y'$, $U'=X'(aba)^2Y'$ whence $U\sim U'$. %
A single reduction $bbabb\to bb$ is examined symmetrically.

Now suppose that $U$ contains no factors $aabaa$ and $bbabb$. %
Without loss of generality assume that $U=Xa(ab)^kaaY$, where $k\geq2$, %
and let $U'$ be obtained from $U$ by the reduction $a(ab)^kaa\to aa$, that is, $U'=XaaY$.
Since $U$ is an~$\widetilde{A}$-whole word, we can write $U = X'aba(ab)^kaabaY'$, %
where $X=X'ab$ and $Y=baY'$. If $X'=\lambda$ or $Y'=\lambda$, then $U$ has a non-reducible left %
(resp., right) tail, which is impossible by the conditions of the proposition.
Hence, $X'\neq\lambda$ and $Y'\neq\lambda$. %
Since the word $U$ contains no factors $aabaa$ and $bbabb$, we have $U=X''baba(ab)^kaababY''$. Then
\begin{equation*}
\begin{split}
U=X''(baba)(ab)^ka(abab)Y''&\sim X''(ba)^{k+1}(ab)^ka(abab)^{k+1}Y'' = \\ %
X''b\,(ab)^ka\,(ab)^ka\,(ab)^ka\,bY''&\sim X''b(ab)^ka\,(ab)^kabY''=\\
X''(ba)^{k+1}(ab)^{k+1}Y''&\sim X''babaababY''=U'.%
\end{split}
\end{equation*}
We get $U\sim U'$, as desired. So, the forward implication is proved.

\medskip The backward implication ($r(U)\sim U$ implies that $U$ has no non-reducible tails) %
trivially follows from Lemma~\ref{RedTails}, since $r(U)$ has no non-reducible tails.
\end{proof}

Our prime interest is in the study of equivalence classes that contain %
almost overlap-free words. Since almost overlap-free words have no  non-reducible tails, %
all words from such equivalence classes have no non-reducible tails by Lemma~\ref{RedTails}. %
So, if a word $U$ is $\widetilde{AB}$-whole, we can replace $U$ by %
the word $r(U)$, which is equivalent to $U$ by Proposition~\ref{RedSave}, and thus significantly %
simplify further analysis. The case when a word $U$ is not $\widetilde{AB}$-whole will be %
considered below in Subsection~\ref{NonUnAOF}.


\subsection{Uniform neighbours and quasi-neighbours.}
\label{Quazi}
This subsection is devoted to the study of uniform neighbours.  %
We already know that $\sim_{r_1} = \pi_{r_1}^+$. %
Let us denote the restrictions of the relations $\sim$ and $\pi$ to the set of all uniform words %
by $\sim_r$ and $\pi_r$ respectively. In this subsection we prove the following equality:

\begin{proposition}
\label{Pi_r} %
$\sim_r=\pi_r^+$. %
\end{proposition}

First, we give more definitions. We say that words $U$ and $V$ %
are \textit{$ab$-neighbours} and write $(U,V)\in\pi_{ab}$ if one of them has %
the form $X(ab)^2Z$ and the other has the form $X(ab)^3Z$ %
for some $X, Z\in\Sigma^*$. Thus we get the \textit{$ab$-neighbourhood} %
relation, which obviously preserves the uniformity. %
We call words $U$ and $V$ \textit{quasi-neighbours} if %
$U=V$ or there exist two sequences $\{U_i\}_{i=1}^n$ and $\{V_j\}_{j=1}^m$ %
of words such that $U_1=U$, $(U_i, U_{i+1})\in\pi_{ab}$ for each %
$i=1, \dotsc, n{-}1$; $V_1=V$, $(V_j, V_{j+1})\in\pi_{ab}$ %
for each $j=1, \dotsc, m{-}1$; and  $(U_n, V_m)\in\pi$. %

The proof of Proposition~\ref{Pi_r} is based on the following lemma.

\begin{lemma}
\label{Quasi_Pi} %
Let $\widetilde{AB}$-whole words $U$ and $V$ have no non-reducible tails. If $U$ and $V$ are neighbours, %
then the words $r(U)$ and $r(V)$ are quasi-neighbours.
\end{lemma}

\begin{proof}
Without loss of generality let $U=XYYZ$ and $V=XYYYZ$ %
for some $X, Z\in\Sigma^*$ and $Y\in\Sigma^+$. %
First, we apply the $r$-reduction to the words $XY[1]$, $Y$, and $Y[|Y|]Z$. %
Obviously, $r(XY[1])=X'Y[1]$ and $r(Y[|Y|]Z)=Y[|Y|]Z'$ for %
some $X', Z'\in\Sigma^*$, and $r(Y)$ begins with $Y[1]$ and %
ends by $Y[|Y|]$. Thus we obtain the neighbours $X'r(Y)r(Y)Z'$ and %
$X'r(Y)r(Y)r(Y)Z'$. Clearly, these words are $\widetilde{AB}$-whole %
and have no non-reducible tails. So, to simplify notation we may assume  that %
the words $XY[1]$, $Y$, and $Y[|Y|]Z$ are already $r$-reduced. Suppose that %
at least one of the words $U$ and $V$ is not $r$-reduced, and a word $P$ %
is its factor of the form $a\widetilde{A}a$ or $b\widetilde{B}b$. %
Without loss of generality we assume that $P=a(ab)^kaa$ for some integer $k>0$.
The word $P$ either contains %
$Y$ as internal factor or occurs inside the factors %
$XY$, $YY$, or $YZ$. Consider first the cases when $Y\ll P$ or $P\leq YY$. %

Suppose $Y\ll P$. Then the word $Y$ is letter-alternating. %
If $Y$ has odd length, then the word $V=XYYYZ$ contains the factor %
$Y[|Y|]YY[1]$ of the form $a\widetilde{A}a$ or $b\widetilde{B}b$.  %
The reduction of this factor turns $V$ into $U$, and the lemma readily follows. %
If the length of $Y$ is even, then the words $YY$ and %
$YYY$ are also letter-alternating. Hence the factor $P$ begins %
inside $X$ and ends inside  $Z$, that is, $P=X'Y^kZ'$, where $X'$ is %
a~suffix of $X$, $Z'$ is a~prefix of $Z$, and either $k=2$ and $P\leq U$ %
or $k=3$ and $P\leq V$. %
Clearly, the words $X'YYZ'$ and $X'YYYZ'$ both have the form %
$a\widetilde{A}a$. Thus we reduce the factor $X'YYZ'$ in $U$ and the factor $X'YYYZ'$ in $V$ %
to get two identical words and prove the lemma.%

In the case $P\leq YY$, we have $P = aY_2Y_1a$, where the words $Y_1a$ and $aY_2$ %
are respectively a prefix and  a suffix of $Y$ such that $Y_2Y_1 = (ab)^ka$. %
Obviously, the prefix $Y_1$ of the word $Y$ %
does not overlap the suffix $Y_2$ of $Y$ (otherwise $Y_1$ contains %
the factor $aY_2[1] = aa$, which is impossible). Thus, $Y=Y_1Y'Y_2$ for %
some $Y'\in\Sigma^*$. We have \[U=XY_1Y'(Y_2Y_1)Y'Y_2Z \text{ and }  %
V=XY_1Y'(Y_2Y_1)Y'(Y_2Y_1)Y'Y_2Z, \] and after the reduction $P\to aa$ we %
get the neighbours $XY_1Y'Y'Y_2Z$ and $XY_1Y'Y'Y'Y_2Z$.

Now suppose that the word $YY$ is already $r$-reduced and %
$Y$ is not an internal factor of $P$. Then $P$  occurs into both words $U$ and $V$ %
inside the factor $XY$ or $YZ$. First, assume  %
that $YZ$ is $r$-reduced. Then $P\leq XY$ and $P$ is a %
unique factor of the form $a\widetilde{A}a$ or $b\widetilde{B}b$ in %
$U$ and $V$. So, $P=aX_1Y_1a$, where $aX_1$ is a suffix of $X$, %
$Y_1a$ is a prefix of $Y$, $X_1\in\Sigma^*$, $Y_1\in\Sigma^+$, and %
$X_1Y_1=(ab)^ka$. Let $Y = Y_1aY'$ for some $Y'\in\Sigma^*$. %
Since the words $U$ and $V$ are $\widetilde{AB}$-whole, we %
have $X = X'abaX_1$ for some $X'\in\Sigma^*$. %
If $k=1$, then \[U=X'aba(aba)aY'YZ \text{ and } V=X'aba(aba)aY'YYZ.\] %
After the reduction $P\to aa$ we obtain the neighbours %
\[X'\underbrace{aba}_{X_1Y_1}aY'YZ = X'X_1\underbrace{Y_1aY'}_YYZ = X'X_1YYZ\]  and  %
\[X'\underbrace{aba}_{X_1Y_1}aY'YYZ = X'X_1\underbrace{Y_1aY'}_YYYZ = X'X_1YYYZ.\] %
In the case $k\geq 2$, we have $X'\neq\lambda$ (if $X'=\lambda$, then both words $U$ and $V$  %
have the non-reducible tail $(aba)(ab)^kaa$, in contradiction with the lemma's condition). %
Since $X$ is $r$-reduced, the last letter of $X'$ cannot be equal %
to $a$, therefore $X=X''babaX_1$. After the reduction $P\to aa$ we %
obtain the quasi-neighbours %
\[U'=X''babaaY'YZ \text{ and } V' = X''babaaY'YYZ.\] %
Indeed, if we put \[U_i = X''ba(ba)^iaY'YZ \text{\, and\, } %
V_i = X''ba(ba)^iaY'YYZ\] %
for $i=1, \dotsc, k$, we  get $U_1=U'$, $V_1=V'$, $(U_i, U_{i+1})\in\pi_{ab}$, %
$(V_i, V_{i+1})\in\pi_{ab}$ for each $i=1, \dotsc, k{-}1$, and %
the words \[U_k = X''b\underbrace{a(ba)^k}_{X_1Y_1}aY'YZ = X''bX_1\underbrace{Y_1aY'}_YYZ = %
X''bX_1YYZ\] and \[V_k =  X''b\underbrace{a(ba)^k}_{X_1Y_1}aY'YYZ = %
X''bX_1\underbrace{Y_1aY'}_YYYZ = X''bX_1YYYZ\] are neighbours. %
So, since $U'=r(U)$ and $V' = r(V)$, the words $r(U)$ and $r(V)$ are quasi-neighbours, as desired.

The case $P\leq YZ$ is considered in the same way. Finally, %
if  both words $XY$ and $YZ$ are not $r$-reduced, %
then there exist two words $P_1$ and $P_2$ of the form %
$a\widetilde{A}a$ or $b\widetilde{B}b$ such that $P_1\leq XY$ %
and $P_2\leq YZ$. In this case, after the reductions of $P_1$ and $P_2$ %
we obtain the words $r(U)$ and $r(V)$, which appear to be %
quasi-neighbours. Indeed, we can construct two sequences of words, %
instead of the one in the previous case, in the same way, increasing degree of the factor %
$ab$ or $ba$ in the prefix $r(XY)$ and in the suffix $r(YZ)$ of the words $r(U)$ and %
$r(V)$. These sequences obviously satisfy all conditions %
from the definition of quasi-neighbours. %
This completes the proof. %
\end{proof}

We should mention that a weaker version of Lemma~\ref{Quasi_Pi} for the case when $U$ and $V$ %
are equivalent to some $\varphi$-images was proved in~\cite{Bakirov}. \medskip%

\begin{proof}[\indent \textup{\textbf{Proof of Proposition~\ref{Pi_r}}}]
We say that a sequence  $\{R_i\}_{i=0}^n$ of words is called %
a \textit{linking $(U, V)$-sequence} if $R_0=U$, $R_n=V$, and $(R_{i-1}, R_i)\in\pi$ %
for each $k=1, \dotsc, n$. If all the words $R_i$ are $r_1$-reduced ($r$-reduced), we call such %
a sequence an \textit{$r_1$-linking} (resp., an \textit{$r$-linking}) \textit{$(U, V)$-sequence}. %
In these terms, Proposition~\ref{Pi_r} means that two uniform words are equivalent %
if and only if there exists an $r$-linking $(U, V)$-sequence. %

If there exists an $r$-linking $(U, V)$-sequence, then $U\sim V$ and both words $U$ and %
$V$ are uniform. Conversely, let $U$ and $V$ be uniform words and $U\sim V$. %
Then there exists an $r_1$-linking  $(U, V)$-sequence $\{W_k\}_{k=0}^n$. %
Since the words $U$ and $V$ are uniform, all words %
$W_k$ are $\widetilde{AB}$-whole and have no non-reducible tails. %
Consider the sequence $\{W'_k = r(W_k)\}_{k=0}^n$. We have
$W'_0=U$, $W'_n=V$,  and, by Lemma~\ref{Quasi_Pi}, %
the words $W'_{k-1}$ and $W'_k$ are uniform quasi-neighbours for each $k = 1, \dotsc, n$. %
So, the sequence $\{W'_k\}_{k=0}^n$ is an $r$-linking $(U, V)$-sequence, %
as desired.
\end{proof}


\subsection{Non-uniform almost overlap-free words. Non-uniform tails.}
\label{NonUnAOF}%
All results mentioned in this subsection were proved in~\cite{MyWork}. %
In the sequel, we often make use of the negation operation $\overline{\phantom{A}}$, %
which is a unique nontrivial automorphism of $\Sigma^*$. For example,
$\varphi(\overline{W})=\overline{\varphi(W)}$. We also write %
$\overline{\mathcal{L}}=\{\overline{W}\mid W\in\mathcal{L}\}$ for any $\mathcal{L}\subset\Sigma^*$.
\medskip

Suppose that a word $U$ is not %
$\widetilde{AB}$-whole and an almost overlap-free word $V$ is equivalent to $U$. %
Then $V$ is not $\widetilde{AB}$-whole as well by Proposition~\ref{WholeEqv}; %
in particular, $V$ is non-uniform. Define
\begin{multline*}
\mathcal{A}_1=\{aabaa, aabaab, baabaa, baabaab, aabaabb, \\ bbaabaa, aabaaba, abaabaa, aabaabbaabaa\},
\end{multline*}
and let $\mathcal{S}_1 = \mathcal{A}_1\cup\overline{\mathcal{A}_1}$. %
The non-uniform almost overlap-free words can be characterized as follows.

\begin{proposition}
\label{NonReg} %
Let $V$ be  a non-uniform $r_1$-reduced almost overlap-free word. %
Then at least one of the following conditions holds:

\textup{1)} $V\in\mathcal{S}_1$; %

\textup{2)} Up to negation, $V$ has the prefix $aabaabba$  or the suffix $abbaabaa$. %
\end{proposition}

Since all words from $\mathcal{S}_1$ are almost overlap-free, not uniform, %
and not $\widetilde{AB}$-whole, we immediately get the following corollary of Proposition~\ref{NonReg}.

\begin{corollary}
\label{ABUn} %
Any almost overlap-free word is $\widetilde{AB}$-whole if and only if it is uniform.
\end{corollary}

The $r_1$-reduced words that are equivalent to one of the prefixes or suffixes mentioned in %
Proposition~\ref{NonReg}, 2) will be called \textit{non-uniform tails}. %
There are four kinds of such tails:\\[3pt]
\centerline{
\begin{tabular}{rcc}
&left&right\\
$A$-tail&$(aab)^*(aab)^2ba$&$ab(baa)^2(baa)^*$\\
$B$-tail&$(bba)^*(bba)^2ab$&$ba(abb)^2(abb)^*$
\end{tabular}
}\\[2pt]

The \textit{tail reduction\/} operation $r_T$, defined in~\cite{MyWork}, %
reduces a left non-uniform tail of any word  to the last 7 symbols, and a right non-uniform %
tail to the first 7 symbols:\\[3pt]
\centerline{
\begin{tabular}{rcc}
&left&right\\
reduced $A$-tail&$abaabba$&$abbaaba$\\
reduced $B$-tail&$babbaab$&$baabbab$
\end{tabular}
}\\[2pt]

Note that a tail of an almost overlap-free word, if any, has  %
exactly 8 symbols. Hence the operation $r_T$ deletes exactly one letter from such a tail.

In addition to the function $r_T$, we use the functions $r^l_T$ and $r^r_T$, which reduce %
a left (resp., right) non-uniform tail of any given word to the last (resp., first) 7 symbols. %
If a word $W$ has both left and right tails, then these tails have at most 6 symbols in common %
(up to negation, for the word $W = (aab)^*aabaabbabb(abb)^*$). Hence %
the operation $r^l_T$ preserves the right tail of $W$ and $r^r_T$ preserves the left tail %
of $W$. Thus we have $r_T(W)=r^r_T(r^l_T(W)) = r^l_T(r^r_T(W))$ for any word $W$.\medskip

We make one easy observation.

\begin{observation}
\label{TailRedInv} %
 Let $U$ and $V$ be $r_1$-reduced words and $r_T(U)\sim r_T(V)$. Then $U\sim V$.
\end{observation}

Note that the inverse is not true in general. %
Surprisingly, under certain conditions, the operation $r_T$ %
preserves the congruence $\sim$.  %

\begin{proposition}
\label{TailReg} Suppose that $U\sim V$ for $r_1$-reduced words $U$ and $V$. Then

\textup{1)}\ $U$ and $V$ have non-uniform tails of the same kind, if any;

\textup{2)}\ If $U, V\not\in [aabaabbaabaa]_{r_1}\cup[bbabbaabbabb]_{r_1}$ and %
at least one of the words $r_T(U)$ and $r_T(V)$ is $\widetilde{AB}$-whole, then $r_T(U)\sim r_T(V)$.
\end{proposition}

Combining Propositions~\ref{TailReg},~\ref{NonReg}, \ref{WholeEqv}, and  %
the definition of $r_T$, we get the following corollary.

\begin{corollary}
\label{TailAOF} %
Let $U\sim V$ for  an $r_1$-reduced word $U$ and an $r_1$-reduced almost overlap-free word $V$. %
If $V\not\in\mathcal{S}_1$, then $r_T(U)\sim r_T(V)$ and the words %
$r_T(U)$, $r_T(V)$ are $\widetilde{AB}$-whole.
\end{corollary}

So, it remains to calculate the equivalence classes of all words from $\mathcal{S}_1$. %
The following lemma gives the answer.

\begin{lemma}
\label{A1} \textup{1)}\ $[aabaa]_{r_1}=aabaa$;

\textup{2)}\ $[aabaab]_{r_1}=(aab)^2(aab)^*$;

\textup{3)}\ $[baabaa]_{r_1}=(baa)^*(baa)^2$;

\textup{4)}\ $[baabaab]_{r_1} = (baa)^2(baa)^*b$;

\textup{5)}\ $[aabaabb]_{r_1} = (aab)^2(aab)^*b$;

\textup{6)}\ $[bbaabaa]_{r_1}=b(baa)^*(baa)^2$;

\textup{7)}\ $[aabaaba]_{r_1}=(aab)^2(aab)^*a$;

\textup{8)}\ $[abaabaa]_{r_1}=a(baa)^*(baa)^2$;

\textup{9)}\ $[aabaabbaabaa]_{r_1} = (aab)^2(aab)^*(b(aab)^*aab)^*(baa)^*(baa)^2$.

The classes of the words from $\overline{\mathcal{A}_1}$ are negations %
of the classes 1)--9).
\end{lemma}

Note that for any word $V\in\mathcal{S}_1$ the equivalence class $[V]_{r_1}$ is %
a regular language. Thus the word problem is decidable for any pair $(U, V)$: %
one can build a finite automaton recognizing $[V]_{r_1}$ and try to accept $U$ by this automaton.

\subsection{From uniform words to $\varphi$-images: functions $\xi$ and $\eta$.}
\label{XiEta} %
Combining the functions $r_T$ and $r$, we can transform a given word $U$ to a uniform %
word. Here we introduce the functions $\xi$ and $\eta$, which transform any uniform word %
to a $\varphi$-image. By Proposition~\ref{RegRed}, any uniform word $U$ is a factor of a $\varphi$-image. %
That is, %
\[U=cQ_1\dotso Q_kd,\text{ where } Q_1,\dotsc,Q_k\in\{ab,ba\}, c,d\in\{a,b,\lambda\}.\] %
This representation is unique %
if $U$ is not letter-alternating. For letter-alternating words we additionally %
require $c=\lambda$ to get a unique representation as well. Now put %
\[\eta(U)=Q_1\dotso Q_k, \quad \xi(U)=\overline{\mathstrut c}U\overline{\mathstrut d}.\] %
Obviously, $\eta(U)$ is a $\varphi$-image of maximum length contained in $U$, while %
$\xi(U)$ is a $\varphi$-image of minimum length containing $U$. We denote $h_\eta(U)=c$, %
$t_\eta(U)=d$, $h_\xi(U)=\overline{c}$, and $t_\xi(U)=\overline{d}$. %
First, we establish some basic properties %
of $\xi$ and $\eta$.



\begin{lemma}
\label{UniqueXiEta}
\textup{1)} Let $eUf\in\varphi(\Sigma^*)$, where $e,f\in\{a,b,\lambda\}$ and $U$ is not %
letter-alternating. Then $\xi(U)=eUf$, $e=h_\xi(U)$, and $f=t_\xi(U)$.

\textup{2)} Let $U=cU'd$, where $c,d\in\{a,b,\lambda\}$ and $U'\in\varphi(\Sigma^*)$. %
Then $U$ is uniform. Moreover, if $U$ is not letter-alternating, then %
$\eta(U)=U'$, $h_\eta(U)=c$,  and $t_\eta(U)=d$;

\textup{3)} Let $U=cU'd$, where the word $U'$ is uniform and %
$c,d\in\{a,b,\lambda\}$ such that $c\neq U[1]$, $d\neq U[|U|]$. %
Then the word $U$ is uniform as well.
\end{lemma}

\begin{proof}
To prove 2), note that $U=cU'd$ and $U'\in\varphi(\Sigma^*)$ imply %
$U\leq\overline {\mathstrut c}cU'd\overline {\mathstrut d}\in\varphi(\Sigma^*)$. Hence %
the word $U$ is uniform. If $U$ is not letter-alternating, then the representation $cU'd$ is %
unique, so we are done. Statement 3 of the lemma is evident. Indeed, from $c\neq U'[1]$ and %
$d\neq U'[|U'|]$ it follows that all factors $aa$ or $bb$ occur in $U$ inside the factor $U'$ %
only. Adding the symbol $c$ to the beginning of $U'$, we change (if $c\neq\lambda$) or save (if $c=\lambda$) %
parity of all positions in the word $U'$. Therefore, since $U'$ is uniform, the word $U$ is uniform as well. %
Finally, one can see that if $eUf=\varphi(X)$ for some $X\in\Sigma^*$, then %
$U=\overline {\mathstrut e}\varphi(X')\overline {\mathstrut f}$, where $X'\leq X$. In view of the uniqueness of %
such representation, we conclude that $\eta(U)=\varphi(X')$, $h_\eta(U)=\overline e$, and %
$t_\eta(U)=\overline f$. From  the definitions of $h_\xi$ and $t_\xi$ it immediately follows that %
$h_\xi(U)=e$ and $t_\xi(U)=f$. The proof is complete.
\end{proof}

The following observation describes the equivalence classes of all letter-alternating words.

\begin{observation}
\label{letteralt} %
The equivalence classes containing letter-alternating words are: %
$[a]=a$, $[ab]=ab$, $[aba]=aba$, $[abab]=(ab)^2(ab)^*$, $[ababa]=(ab)^2(ab)^*a$, %
and their negations.
\end{observation}

Now we prove that the functions $h_\xi$, $t_\xi$, $h_\eta$, and $t_\eta$ %
are invariant under the congruence $\sim$.

\begin{proposition}
\label{HTXiEta} %
Let $U,V$ be uniform words and  $U\sim V$. Then

\textup{1)}\quad  $h_{\xi}(U) = h_{\xi}(V),\quad t_{\xi}(U)=t_{\xi}(V)$;

\textup{2)}\quad $h_{\eta}(U) = h_{\eta}(V),\quad t_{\eta}(U) = t_{\eta}(V)$.
\end{proposition}

\begin{proof}
If $U$ or $V$ is letter-alternating, then the proposition readily follows from %
Observation~\ref{letteralt}. Assume that they are not letter-alternating.

Let us prove the first statement. By definition of congruence, $U\sim V$ implies %
$h_\xi(U)Ut_\xi(U)\sim h_\xi(U)Vt_\xi(U)$. From $U\sim V$ and %
the definition of $\xi$ it follows that $h_\xi(U)\neq U[1]=V[1]$ and $t_\xi(U)\neq U[|U|]=V[|V|]$. %
Hence, by Lemma~\ref{UniqueXiEta}, 3), the word $h_\xi(U)Vt_\xi(U)$ is uniform.
So we have $\xi(U)\sim h_\xi(U)Vt_\xi(U)$, where $\xi(U)\in\varphi(\Sigma^*)$ and %
the word $h_\xi(U)Vt_\xi(U)$ is uniform. According to Proposition~\ref{MainProp}, %
we get $h_\xi(U)Vt_\xi(U)\in\varphi(\Sigma^*)$ as well. %
In view of Lemma~\ref{UniqueXiEta}, 1), we have $h_\xi(V)=h_\xi(U)$ %
and $t_\xi(V)=t_\xi(U)$, as desired.

Statement 2 is trivially follows from statement 1 and the definitions %
of $h_\eta$, $t_\eta$, $h_\xi$, and $t_\xi$.
\end{proof}

\medskip
According to Proposition~\ref{HTXiEta}, the function $\xi$ preserves the congruence $\sim$ %
whereas the function $\eta$ does not preserve $\sim$ in general case. However, as we will show %
in Sect.~\ref{Zloo}, under certain conditions $\eta$ preserves $\sim$ as well. %
The function $\xi$ is used in the proof of Theorem~\ref{Th1}, while %
the function $\eta$ plays a crucial role in the construction of Algorithm EqAOF.
\medskip

Finally, we note that the functions $\varphi^{-1}(\xi(V))$ and $\varphi^{-1}(\eta(V))$ preserve %
the property of a word to be almost overlap-free. For $\xi$ this was proved in~\cite{Shur}, %
for $\eta$ even a stronger assertion holds. %

\begin{observation}
\label{phi_minus_one} If $V$ is almost overlap-free, then the word $\varphi^{-1}(\eta(V))$ %
is overlap-free. %
\end{observation}

\begin{proof}
An almost overlap-free word that is not overlap-free has the form $cYcYc$ for %
some $c\in\Sigma$ and $Y\in\Sigma^*$, hence, it has  odd length. %
Since the length of $\eta(V)$ is even, $\eta(V)$ is a proper factor of $V$ and, hence, it is %
overlap-free. The function $\varphi^{-1}$ preserves overlap-freeness, so the observation follows.
\end{proof}

\section{Proof of Theorem~\ref{Th1}}
\label{ProofOne} %
We prove Theorem 1 in the form~$(*)$ and use the minimal counterexample method.
Suppose that a pair $(U,V)$ %
provides a counterexample, i. e.,  $U$ and $V$ are two nonequal $r_1$-reduced  %
almost overlap-free words such that $U\sim V$, %
and $l=\min\{|U|,|V|\}$ takes the minimum value among all such pairs.  %
We aim to get a contradiction by obtaining a ``shorter'' counterexample.

Clearly, $l>2$. By Lemma~\ref{A1}, any equivalence class %
$[W]_{r_1}$ for $W\in\mathcal{S}_1$ contains exactly one %
almost overlap-free word (namely, the word $W$ itself). %
Therefore $U,V\not\in\mathcal{S}_1$. Similarly, by %
Observation~\ref{letteralt}, the words $U$ and $V$ are not letter-alternating. %
It follows from Corollaries~\ref{ABUn} and~\ref{TailAOF}  that %
$r_T(U)\sim r_T(V)$ and $r_T(U)$, $r_T(V)$ are uniform almost overlap-free words. %
Hence, by Proposition~\ref{HTXiEta}, the words $U_1 =\varphi^{-1}(\xi(r_T(U)))$ and
$V_1=\varphi^{-1}(\xi(r_T(V)))$ are equivalent. %
Obviously, both words $U_1$ and $V_1$ are almost overlap-free and $U_1\neq V_1$. %

Since the words $U$ and $V$ are not letter-alternating, the %
words $\xi(r_T(U))$ and $\xi(r_T(V))$ are not letter-alternating as well. %
In particular, both words $\xi(r_T(U))$ and $\xi(r_T(V))$ are not equal %
to $ababab$ or $bababa$. Therefore $U_1, V_1\not\in\{aaa,bbb\}$. %

Finally, we have $l_1=\min\{|U_1|,|V_1|\}\leq (l+2)/2<l$ whenever $l>2$. %
So the pair $(U_1, V_1)$ provides a shorter counterexample. This completes the proof of Theorem~\ref{Th1}. %

\section{Procedure Ancestor}
\label{Ancestor} %
In this section we begin the construction of Algorithm AqEOF, which returns the almost %
overlap-free word that is equivalent to a given word or reports %
that no such almost overlap-free word exists. %
Algorithm AqEOF includes two main procedures. Let us define the first of them.

\noindent\textbf{Procedure Ancestor}. \\
\textit{Input}.\ A word $U\in\Sigma^+$. \\
\textit{Output}.\ A word $\mathrm{Anc}(U)\in\Sigma^+$, integer $k$, %
length $k$ arrays $L$, $R$, $h$, $t$ of letters. \\[3pt]
\textbf{Step\ 0.}\  Let $k:=0$. \\
\textbf{Step\ 1.}\  Let $U=r_1(U)$; $k=k+1$, $L[k]=R[k]=h[k]=t[k]:=\lambda$.\\
\textbf{Step\ 2.}\  If $|U|\leq2$ or $U\in[aabaabbaabaa]\cup[bbabbaabbabb]$,
then $\mathrm{Anc}:=U$; stop.\\
\textbf{Step\ 3.}\  If $U$ has a non-uniform left (right) tail, %
set $L[k]:=U[1]$ (resp., $R[k]=U[|U|]$); %
let $U':=r_T(U)$.\\
\textbf{Step\ 4.}\  If $U'$ is not  $\widetilde{AB}$-whole %
or $U'$ has a non-reducible tail, then $\mathrm{Anc}:=U$; stop.\\
\textbf{Step\ 5.}\  Let $U':=r(U')$.\\
\textbf{Step\ 6.}\  Let $h[k]:=h_\eta(U')$; $t[k]:=t_\eta(U')$;  %
$U:=\varphi^{-1}(\eta(U'))$; goto step 1. \\
\textbf{End.}\medskip

Starting with $U_1=r_1(U)$, procedure Ancestor constructs %
the sequence of words $U_1, U_2, \dotsc$ by the rule: %
\[U_{k+1} =r_1(\varphi^{-1}(\eta(r(r_T(U_{k-1})))))\]
until one of the stop conditions is fulfilled.  %
We call the sequence $\{U_k\}$ the \textit{primary $U$-series}. %
Since $|U_k|<|U_{k-1}|/2$ for any $k\geq2$, the primary $U$-series %
is finite, its \textit{length} (that is, the number of words in it) %
is denoted by $\ell(U)$. We say that the output word %
$\Anc(U)= U_{\ell(U)}$ is the \textit{ancestor of $U$}, and the arrays $L=L_U$, $R=R_U$, $h=h_U$, %
and $t=t_U$ returned by Procedure Ancestor are %
\textit{associated with~$U$}.  We omit the index $U$ if it is clear from context. %

The next two lemmas establish the basic properties of  primary series.  %
First, we examine primary series of almost overlap-free words.

\begin{lemma}
\label{Anc1} %
Let $V$ be an almost overlap-free word and  %
$\{V_k\}_{k=1}^{\ell(V)}$ be its primary $V$-series. Then

\textup{1)}\quad $\Anc(V)\in\mathcal{S}_1\cup\{a, b, ab, ba, aa, bb\}$; %

\textup{2)}\quad $V_k$ is overlap-free for each $k=2, \dotsc, \ell(V)$; %

\textup{3)}\quad $V_k=L[k]r(r_T(V_k))R[k]$ for each $k=1, \dotsc, \ell(V){-}1$.
\end{lemma}

\begin{proof}
Instead of 1)--3) we prove the following statement:  %
\begin{itemize}%
\item[] \textit{if for some $k\geq1$ the word $V_k$ is almost overlap-free, $|V_k|>2$, and $V_k\not\in\mathcal{S}_1$, %
then $V_k=L[k]r(r_T(V_k))R[k]$, $k < \ell(V)$, and the word $V_{k+1}$ is overlap-free.} %
\end{itemize}
Consider the $k$th iteration of procedure Ancestor, i. e., the processing of the word $V_k$. %
By the conditions above, procedure Ancestor cannot stop on Step 2. By Proposition~\ref{NonReg}, %
$r_T(V_k)$ is a  uniform %
almost overlap-free word. Hence procedure Ancestor cannot stop on Step 4 as well. %
Thus procedure Ancestor does not stop on the $k$th iteration, so we have $k<\ell(V)$,
and the word $V_{k+1}=r_1(\varphi^{-1}(\eta(r_T(V_k))))$ is overlap-free %
in view of Observation~\ref{phi_minus_one}.  %
From the uniformity of $r_T(V_k)$ it follows that $r(r_T(V_k))=r_T(V_k)$. %
Since the function $r_T$ deletes at most one letter from the beginning and at most %
one letter from the end of any almost overlap-free word (see the remark after the definition %
of $r_T$), we get $V_k=L[k]r(r_T(V_k))R[k]$, as desired.
\end{proof}

Now we apply procedure Ancestor to a pair $(U, V)$ of equivalent words. %

\begin{lemma}
\label{Anc2} %
Suppose that $U\sim V$, $m= \min\{\ell(U),\ell(V)\}$,  $\{U_k\}_{k=1}^{\ell(U)}$, %
$\{V_k\}_{k=1}^{\ell(V)}$ are the primary $U$- and $V$-series  respectively, %
and  the following condition: %
\[r(r_T(U_k)))\sim r(r_T(V_k)))\Rightarrow \eta(r(r_T(U_k)))\sim\eta(r(r_T(V_k)))%
\text{ for all } k<m\leqno{(**)}\]
holds. Then %

\textup{1)}\quad $\ell(U)=\ell(V)$;

\textup{2)}\quad $U_k\sim V_k$ for each $k=1,\dotsc, m$; in particular, $\Anc(U)\sim\Anc(V)$;

\textup{3)}\quad $L_U=L_V$, $R_U=R_V$, $h_U=h_V$, and $t_U=t_V$.
\end{lemma}

\begin{proof}
Instead of 1)--3) we prove the following statement: %
\begin{itemize}
\item[] \textit{if $U_k\sim V_k$ for some $k\leq m$, then %
$L_U[k]=L_V[k]$, $R_U[k]=R_V[k]$, $h_U[k]=h_V[k]$, $t_U[k]=t_V[k]$, %
and either $k=m=\ell(U)=\ell(V)$ or $k<m$ and $U_{k+1}\sim V_{k+1}$.} %
\end{itemize}
First, suppose that  $k=m=\ell(U)$. %
If procedure Ancestor stops to process the word $U_k$ on Step 2, then the %
processing of the word $V_k$ should stop on Step 2 as well, because $U_k\sim V_k$ and %
both words are $r_1$-reduced. %
In this case, we have $\ell(U)=\ell(V)$, %
and the $k$th elements of all arrays associated with $U$ and $V$ are equal to $\lambda$.

Now suppose that procedure Ancestor stops to process the word $U_k$ on Step 4. %
Assume additionally that the word $V_k$ is $\widetilde{AB}$-whole and  has no non-reducible tails. %
Then the word $r_T(V_k)$ is $\widetilde{AB}$-whole as well and we have $r_T(U_k)\sim r_T(V_k)$ by Proposition~\ref{TailReg}, 2). %
At the same time, from Lemma~\ref{RedTails} and Proposition~\ref{WholeEqv} it follows %
that the word $r_T(U_k)$ is $\widetilde{AB}$-whole and has no non-reducible tails %
as well as $r_T(V_k)$. So procedure Ancestor cannot stop to process %
the word $U_k$ on Step 4, a contradiction.  Hence either the word $V_k$ is not $\widetilde{AB}$-whole %
or it has a non-reducible tail. In both cases procedure Ancestor stops to process the word $V_k$ %
on Step 4, and we get $k=\ell(U)=\ell(V)$. Also,  $L_U[k]=L_V[k]$, $R_U[k]=R_V[k]$ by %
Proposition~\ref{TailReg}, 1) and $h_U[k]=t_U[k]=h_V[k]=t_V[k]=\lambda$. %

The case $k=m=\ell(V)$ is symmetrical to above one. Finally, suppose that $k<m$. %
We have $r(r_T(U_k))\sim r(r_T(V_k))$ by Propositions~\ref{RedSave} and~\ref{TailReg}. %
From condition~$(**)$
it now follows %
that $\eta(r(r_T(U_k)))\sim\eta(r(r_T(V_k)))$ whence 
\[U_{k+1} = r_1(\varphi^{-1}(\eta(r(r_T(U_k)))))\sim r_1(\varphi^{-1}(\eta(r(r_T(V_k)))))=V_{k+1}. \]
In additional, we have $L_U[k]=L_V[k]$, $R_U[k]=R_V[k]$ by Proposition~\ref{TailReg}, 1) %
and $h_U[k]=h_V[k]$, $t_U[k]=t_V[k]$ by Proposition~\ref{HTXiEta}. This completes the proof of %
the statement.
\end{proof}\medskip

As we said above, under certain conditions, the function $\eta$ preserves the congruence $\sim$.
In the sequel, we  call a pair $(P, Q)$ of equivalent uniform words \textit{good} %
if $\eta(P)\sim\eta(Q)$, and \textit{bad} otherwise. %
So, condition~$(**)$ provides that all pairs %
$(r(r_T(U_k)), r(r_T(V_k)))$ are good. We set the study of bad pairs aside %
for Sect.~\ref{Zloo}.

\section{Normal series}
\label{Normal}
The notion in the headline plays a crucial role in the second main procedure %
of Algorithm EqAOF (this procedure will be presented in Sect.~\ref{AlgEqAOF}). %
This notion is defined as follows.

Take a word $U$ and a cube-free word $W\in[Anc(U)]$.
\textit{The normal $U_W$-series} or \textit{the $W$-normal series of $U$} is the sequence of words
$\widetilde{U}_{\ell(U)}, \widetilde{U}_{\ell(U)-1}, \dotsc, \widetilde{U}_1$ %
defined by
\begin{equation*}
\widetilde{U}_{\ell(U)}=W,\quad %
\widetilde{U}_k = L_U[k]h_U[k]\varphi(\widetilde{U}_{k+1})t_U[k]R_U[k],\quad k=\ell(U){-}1, \dotsc, 1.
\end{equation*}
The word $\widetilde{U}_1$ is called the \textit{$W$-normal} (or simply a \textit{normal}) %
\textit{form of the word $U$} and denoted by $N_W(U)$. Clearly, any word $U$ has at least one normal %
series, since the class $[\Anc(U)]$ contains cube-free words. %
However, the class $[\Anc(U)]$ can contain several cube-free words, therefore $U$ %
can have several normal series. %
If $W$ is almost overlap-free, then the normal $U_W$-series %
is the \textit{main normal $U$-series}, the $W$-normal form of $U$ %
is called the \textit{main normal form of $U$} and denoted by $N(U)$. %
By Theorem~\ref{Th1}, each word $U$ has at most one main normal series. %
A normal series or main normal series is called \textit{direct} if $W=\Anc(U)$. %
In this case, the normal form $N_W(U)$ is called \textit{direct} as well %
and denoted by $N^D(U)$.

First, we prove the  main property of normal forms.

\begin{lemma}
\label{Norm1} %
Suppose that $W\sim \Anc(U)$ for a word $U$ and a cube-free word $W$. %
Let  $\{U_k\}_{k=1}^{\ell(U)}$ and $\{\widetilde{U}_k\}_{k=1}^{\ell(U)}$ be the primary %
$U$-series and the normal $U_W$-series respectively. %
Then  $\widetilde{U}_k\sim U_k$ for each $k=1, \dotsc, \ell(U)$; in particular, $N_W(U)\sim U$.
\end{lemma}

\begin{proof}
For $k=\ell(U)$, we have $\widetilde{U}_{\ell(U)}=W$ and $U_{\ell(U)}=\Anc(U)$ %
from definitions. Now suppose that  $\widetilde{U}_{k+1}\sim U_{k+1}$ for some $k < \ell(U)$. %
If we show that $\widetilde{U}_k\sim U_k$, the required statement will follow by induction. %
We have %
\[\widetilde{U}_{k+1}\sim U_{k+1}= r_1(\varphi^{-1}(\eta(r(r_T(U_k)))))\] %
whence we get  $\widetilde{U}_{k+1}\sim \varphi^{-1}(\eta(r(r_T(U_k))))$ %
and $h[k]\varphi(\widetilde{U}_{k+1})t[k]\sim r(r_T(U_k))$. %
Since procedure Ancestor does not stop on $k$th iteration, the word $r_T(U_k)$ %
is $\widetilde{AB}$-whole and has no non-reducible tails. %
From Proposition~\ref{RedSave} it now follows that $r(r_T(U_k))\sim r_T(U_k)$. Hence,  %
$h[k]\varphi(\widetilde{U}_{k+1})t[k]\sim r_T(U_k)$. %
Note that if $U_k$ has non-uniform tails, then the word $L[k]r_T(U_k)R[k]$ has %
non-uniform tails of the same kind as the word $U_k$.
Since all non-uniform tails of the same kind are equivalent, %
we conclude that $L_[k]r_T(U_k)R[k]\sim U_k$. %
Thus, we have %
\[\widetilde{U}_k=L[k]h[k]\varphi(\widetilde{U}_{k+1})t[k]R[k]\sim L[k]r_T(U_k)R[k]\sim U_k,\] %
as desired.
\end{proof}

So, normal series are ``inverted''  to primary series in a sense. %
Namely, by~Lemma~\ref{Norm1}, if a cube-free word $W$ is equivalent to $\Anc(U)$, %
then the $W$-normal series of $U$ allows one to restore the primary series %
$\{U_k\}_{k=1}^{\ell(U)}$ up to equivalent words. %
Actually, the $U_W$-normal series consists of the words $\widetilde{U}_k$ with %
more simple structure than $U_k$ in the general case. The following lemma describes %
the structure of the words $\widetilde{U}_k$.

\begin{lemma}
\label{Norm2} %
Suppose that $\Anc(U)\sim W$ for a word $U$ and a cube-free word $W$.  %
Let $\{\widetilde{U}_k\}_{k=0}^{\ell(U)}$ be %
the $W$-normal series of $U$. Then

\textup{1)}\  $\widetilde{U}_k$ is an $r_1$-reduced word for each $k=1,\dotsc, \ell(U)$; %

\textup{2)}\  $r_T(\widetilde{U}_k)$ is uniform for each $k=1,\dotsc, \ell(U){-}1$;

\textup{3)}\ $r_T(\widetilde{U}_k)=h[k]\varphi(\widetilde{U}_{k+1})t[k]$ %
for each $k=1, \dotsc, \ell(U){-}1$; %

\textup{4)} $h(r_T(\widetilde{U}_k))=h[k]$, $t(r_T(\widetilde{U}_k))=t[k]$, %
and $\eta(r_T(\widetilde{U}_k))= \varphi(\widetilde{U}_{k+1})$ %
for each $k=1, \dotsc, \ell(U){-}1$.
\end{lemma}

\begin{proof}
For $k=\ell(U)$, there is nothing to prove. Suppose  $k < \ell(U)$. %
Let us denote the word $h[k]\varphi(\widetilde{U}_{k+1})t[k]$ by $\widetilde{U}'_k$, %
and let $\{U_k\}_{k=1}^{\ell(U)}$ be the primary $U$-series.
The word $\widetilde{U}'_k$ is uniform by Lemma~\ref{UniqueXiEta}, 2). %
Since  $\widetilde{U}_{k+1}\sim U_{k+1}\sim \varphi^{-1}(\eta(r(r_T(U_k))))$ by Lemma~\ref{Norm1}, %
we get $\widetilde{U}'_k\sim r(r_T(U_k))$. From Proposition~\ref{HTXiEta} it now follows that %
$h(\widetilde{U}'_k)=h[k]$, $t(\widetilde{U}'_k)=t[k]$, and %
$\eta(\widetilde{U}'_k)=\varphi(\widetilde{U}_{k+1})$.

It remains to prove that $\widetilde{U}_k$ is an $r_1$-reduced word and %
$r_T(\widetilde{U}_k)=\widetilde{U}'_k$. %
We have $\widetilde{U}_k=L[k]\widetilde{U}'_kR[k]$. %
First, we apply the function $r^l_T$ to the word $L[k]\widetilde{U}'_k$. %
Clearly, if $L[k]=\lambda$, then $r^l_T(L[k]\widetilde{U}'_k) = \widetilde{U}'_k$. %
Now suppose that $L[k]\neq\lambda$. Without loss of generality we assume that $L[k]=a$.  %
This means that $U_k$ has the non-uniform left tail $(aab)^k(aab)^2ba$ for some $k\geq 0$. %
Hence the word $r(r_T(U_k))$ has the prefix $abaabba$.  %
Therefore the word $\widetilde{U}'_k$, which is equivalent to $r(r_T(U_k))$, %
begins with $abaabba$ as well. Since the word $L[k]\widetilde{U}'_k$ does not begin with $aaa$, %
it is $r_1$-reduced. So we obtain $r^l_T(L[k]\widetilde{U}'_k) = \widetilde{U}'_k$. %

In the symmetric way one can prove that $\widetilde{U}'_kR[k]$ is an $r_1$-reduced word %
and $r^r_T(\widetilde{U}'_kR[k]) = \widetilde{U}'_k$. %
So, the word $\widetilde{U}_k = L[k]\widetilde{U}'_kR[k]$ is $r_1$-reduced %
and we have %
\[r_T(\widetilde{U}_k) = r^r_T(r^l_T(L[k]\widetilde{U}'_kR[k])) = r^r_T(\widetilde{U}'_kR[k]) = %
\widetilde{U}'_k. \] %
This completes the proof.
\end{proof}

Next we investigate primary series and normal series for normal forms. %

\begin{lemma}
\label{Norm3} %
Let $\Anc(U)\sim W$ for a word $U$ and a cube-free word $W$. %
Then both the primary $N_W(U)$-series and the $W$-normal series of $N_W(U)$ coincide with %
the $U_W$-normal series.  %
\end{lemma}

\begin{proof}
Let us denote $P=N_W(U)$, $m=\ell(U)$, $l=\ell(P)$, %
and let $\{U_k\}_{k=1}^m$, $\{\widetilde{U}_k\}_{k=1}^m$, and $\{P_k\}_{k=1}^l$ be %
the primary series of $U$, the $U_W$-normal series, and the primary series of $P$ respectively. %
First, we prove that the primary $P$-series coincides with  the normal $U_W$-series.  %

By Lemma~\ref{Norm2}, the word $P$ is $r_1$-reduced. Hence, $P_1=P=N_W(U)=\widetilde{U}_1$. %
Now suppose that $P_k=\widetilde{U}_k$ for some $k<m$. %
In view of $\widetilde{U}_k\sim U_k$ and $k<\ell(U)$, we conclude that %
$|P_k|>2$ and  $P_k\not\in[aabaabbaabaa]\cup[bbabbaabbabb]$. %
In addition, the word $r_T(P_k)$ is uniform by Lemma~\ref{Norm2}. %
Hence procedure Ancestor does not stop while processing the word  $P_k$, that is, $k<l$, and we have %
$r(r_T(P_k))=r_T(\widetilde{U}_k)$. Finally, by Lemma~\ref{Norm2} we get %
\begin{multline*}
P_{k+1} = r_1(\varphi^{-1}(\eta(r(r_T(P_k))))) = r_1(\varphi^{-1}(\eta(r_T(\widetilde{U}_k))))=\\%
r_1(\varphi^{-1}(\varphi(\widetilde{U}_{k+1}))) =r_1(\widetilde{U}_{k+1})= \widetilde{U}_{k+1},
\end{multline*}
as desired.

Now consider the word $P_m=\widetilde{U}_m=W\sim\Anc(U)$. %
One can easily check that if procedure Ancestor stops processing the word $\Anc(U)$ on %
Step 2 or Step 4, then it stops processing the word $P_m$ on Step 2 or respectively %
Step 4 as well. Thus, $l=m$, and the primary $P$-series coincides with the normal $U_W$-series.

Clearly, since $P_k=\widetilde{U}_k\sim U_k$, we have $L_U[k]=L_P[k]$, $R_U[k]=R_P[k]$, %
$h_U[k]=h_P[k]$, and $t_U[k]=t_P[k]$ for each $k=1, \dotsc, m$. %
This implies that the $P_W$-normal series coincides with the $U_W$-normal %
series. The proof is complete.
\end{proof}\bigskip

In particular, it follows from Lemma~\ref{Norm3} that $N_W(N_W(U))=N_W(U)$ for any  %
word $U$ and any cube-free word $W$ such that $W\sim\Anc(U)$. So, the repeated use %
of the primary and normal $U$-series has no effect.\medskip

Finally, we study the main normal series for almost overlap-free words.

\begin{lemma}
\label{Norm4} %
If $U$ is an almost overlap-free word and $U\not\in\{aaa, bbb\}$, then %
there exists the main direct normal $U$-series and it coincides with the primary $U$-series; %
in particular, $N(U)=U$.
\end{lemma}

\begin{proof}
Let $\{U_k\}_{k=1}^{\ell(U)}$ and $\{\widetilde{U}_k\}_{k=1}^{\ell(U)}$ %
be the primary and the direct normal $U$-series respectively.
We prove the lemma by induction on $k=\ell(U), \dotsc, 2, 1$. %
For $k=\ell(U)$, we have $\widetilde{U}_{\ell(U)}=\Anc(U)=U_{\ell(U)}$ by the definition of %
direct series. Note that since the word $\Anc(U)$ is almost overlap-free by Lemma~\ref{Anc1}, %
the direct normal series $\{\widetilde{U}_k\}_{k=1}^{\ell(U)}$ is %
the main direct normal series of $U$.

Now suppose that  $\widetilde{U}_k=U_k$, where $1<k\leq\ell(U)$, and prove %
that $\widetilde{U}_{k-1}=U_{k-1}$. %
Let $U'_k=\eta(r(r_T(U_{k-1})))$.  %
First, we prove that the word $\varphi^{-1}(U'_k)$ is $r_1$-reduced. Indeed, %
if $aaa\leq\varphi^{-1}(U'_k)$ (or $bbb\leq\varphi^{-1}(U'_k)$), %
then $ababab\leq U'_k$ (resp., $bababa\leq U'_k$). %
However, it is impossible, since $U'_k\leq r(r_T(U_{k-1}))\leq U_{k-1}$ by Lemma~\ref{Anc1} %
and $U_{k-1}$ is almost overlap-free. %
Hence the word $\varphi^{-1}(U'_k)$ is $r_1$-reduced and %
$U_k = r_1(\varphi^{-1}(U'_k))=\varphi^{-1}(U'_k)$.  %
Therefore we obtain %
\begin{equation*}
\begin{split}
U_{k-1} &= L[k{-}1]r(r_T(U_{k-1}))R[k{-}1] \quad \text{(by Lemma~\ref{Anc1})} \\
 &=L[k{-}1]\,h[k{-}1]U'_k\,t[k{-}1]R[k{-}1] = L[k{-}1]\,h[k{-}1]\varphi(U_k)\,t[k{-}1]R[k{-}1] \\
 &= L[k{-}1]\,h[k{-}1]\varphi(\widetilde{U}_k)\,t[k{-}1]R[k{-}1] = \widetilde{U}_{k-1},  %
\end{split}
\end{equation*}
as desired.
\end{proof}

Lemmas~\ref{Anc1}, \ref{Anc2}, and~\ref{Norm4} together provide the following assertion: %
\begin{itemize}
\item[]\textit{if a word $U$ is equivalent to an almost overlap-free word $V\not\in\{aaa, bbb\}$ and condition $(**)$ holds, %
then $V=N(U)$}.%
\end{itemize}

Indeed, from Lemma~\ref{Anc2} it follows that  $\Anc(U)\sim\Anc(V)$, $L_U=L_V$, $R_U=R_V$, %
$h_U=h_V$, and $t_U=t_V$. By Lemma~\ref{Anc1}, the word $\Anc(V)$ is almost overlap-free, in particular, %
it is cube-free.  Thus the $\Anc(V)$-normal series of the words $U$ and $V$ coincide.  %
Since the word $\Anc(V)$ is almost overlap-free, these normal series are the main normal series %
of the words $U$ and $V$. Hence, $N(U)=N(V)$. By Lemma~\ref{Norm4}, we get $V=N(V)=N(U)$, as desired.\medskip

We conclude this section explaining why do we use the function $\eta$, not $\xi$, %
for the construction of Algorithm AqEOF. Replacing  $\eta$ by $\xi$ %
in procedure Ancestor, we simplify the formulation of Lemma~\ref{Anc2}, since
condition~$(**)$ can be replaced by the following condition, which is obviously true: %
\[r(r_T(U_k))\sim r(r_T(V_k))\Rightarrow \xi(r(r_T(U_k)))\sim \xi(r(r_T(V_k)))\text{ for all } k < m,\]%
where $m=\min\{\ell(U),\ell(V)\}$. %
Clearly, Lemmas~\ref{Anc1}, \ref{Anc2}, and~\ref{Norm4} hold true after replacing %
$\eta$ by $\xi$. So, if we retain the same notion of normal series, %
we will get the following assertion: %
\begin{itemize}%
\item[] \textit{if a word $U$ is equivalent to an almost %
overlap-free word $V\not\in\{aaa, bbb\}$, then $V=N(U)$}.%
\end{itemize}

Unfortunately, Lemma~\ref{Norm1} fails if we use $\xi$ instead of $\eta$. %
Actually, in this case the word $r_T(\widetilde{U}_k)$ will be obtained from $\varphi(\widetilde{U}_{k+1})$ %
by deleting the symbols $h[k] = h_\xi(r(r_T(U_k)))$ and $t[k]=t_\xi(r(r_T(U_k)))$ added to %
the word $r(r_T(U_k))$ by procedure Ancestor. However, deleting  does not preserve %
the congruence $\sim$ in general case.
For example, consider the word $U=bababb$.  Using $\xi$ instead of $\eta$, %
we get $N(bababb) = babb$. The word $N(U)$ is almost overlap-free, but %
$babb\not\sim bababb$. In fact, the equivalence class $[bababb]$ contains no %
almost overlap-free words. So, we simplify the necessary condition for a given word $U$ to be %
equivalent to an almost overlap-free word, but we lose the sufficient condition provided %
by Lemma~\ref{Norm1}. On the other hand, the use of $\eta$ makes Lemma~\ref{Norm1} true. %
This is the main reason why the function $\eta$, not $\xi$, is used in the construction of Algorithm EqAOF. %
Surprisingly, as we will prove in the next section, under certain restrictions, %
the function $\eta$ preserves the congruence $\sim$.  Moreover, even if %
$U\sim V$ and $\eta(U)\not\sim\eta(V)$ for some pair $(U, V)$ (that is, the pair $(U, V)$ is bad), %
then an analogue of Lemma~\ref{Anc2} holds (see Lemma~\ref{Zlo} below).
\medskip %

\section{Bad pairs of equivalent uniform words}
\label{Zloo} %
This section is devoted to the study of bad pairs, for which $U\sim V$ and %
$\eta(U)\not\sim\eta(V)$.  In the sequel, we write $h(U)$ and $t(U)$ instead of $h_\eta(U)$ and $t_\eta(U)$.

First, we study bad pairs of neighbours. %
More precisely, we consider pairs $(U, V)$ such that $(U, V)\in\pi$ %
and $(\eta(U), \eta(V))\not\in\pi$.

\begin{lemma}
\label{EtaNonSave} %
Suppose $U$ and $V$ are uniform, $|U|<|V|$, $(U, V)\in\pi$, and %
$(\eta(U), \eta(V))\not\in\pi$. Then $\{h(U), t(U)\}=\{a,b\}$ %
and there exists an even-length word $Y$ such that $|Y|>2$, $U=YY$, $V=YYY$, %
$h(Y)=h(U)=h(V)$, and $t(Y)=t(U)=t(V)$.
\end{lemma}

\begin{proof}
Since $|U|<|V|$, we can write  $U=XYYZ$ and $V=XYYYZ$, where %
$X,Z \in\Sigma^*$ and $Y\in\Sigma^+$. %
We show that the word $Y$ is the required one.

First, we prove that the length of $Y$ is even. Assume the converse. %
If the word $Y$ is letter-alternating, then $V$ has the factor $Y[|Y|]YY[1]$ of the form $a\widetilde{A}a$ or %
$b\widetilde{B}b$, a contradiction with the uniformity of $V$. %
Now let $Y=PccQ$ for some $P, Q\in\Sigma^*$, and $c\in\Sigma$.  %
Then we have $ccQPcc\leq YY\leq U$ and $ccQPcc\leq YY\leq V$. %
This contradicts the uniformity of $U$ and $V$ again, since the length of $PQ$ is odd. %
Hence, the length of $Y$ is even.

By Proposition~\ref{HTXiEta}, we have $h(U) = h(V) = c$ and $t(U)=t(V) = d$ for some
$c,d\in\{a,b,\lambda\}$. If $|c|\leq |X|$ and $|d|\leq |Z|$, then %
the words $\eta(U)$ and $\eta(V)$ are neighbours by definition. Hence, at least one of these %
inequalities should fail.%

Suppose that $|d|\leq|Z|$ and $|c| > |X|$, that is, $c\neq\lambda$ and $X=\lambda$. %
Then  $Y=cY'e$ for some $Y'\in\Sigma^*$ and $e\in\Sigma$. %
So, we get $\eta(U) =Y'ecY'eZ'$ and $\eta(V) = Y'ecY'ecY'eZ'$, where $Z'\in\Sigma^*$. %
By Observation~\ref{EvenPhi}, the length of $Z'$ is odd, in particular, $Z'\neq\lambda$. %
Also, we have $Z'[1]=\overline{e}=c$, since the factors $ec$ and %
$eZ'[1]$ start in $\eta(U)$ %
at odd positions. Thus the words $\eta(U) = (Y'ec)^2Z''$ and %
$\eta(V) = (Y'ec)^3Z''$, where $Z'' = Z'[2\dotso|Z'|]$, are neighbours, a contradiction.

The case $|c|\leq|X|$ and $|d| > |Z|$ is symmetric to the above one. %
Thus, $X=Z=\lambda$ and $c, d\neq\lambda$. Then $U=YY$, $V=YYY$,  $Y[1]=c$, and $Y[|Y|]=d$.  %
Let  $Y = cY'd$, where $Y'\in\Sigma^*$. %
Then  $\eta(U) = Y'dcY'$ and $\eta(V) = Y'dcY'dcY'$. By Observation~\ref{EvenPhi}, %
we have $Y'\in\varphi(\Sigma^*)$ and $c=\overline d$. If the word $Y$ is letter-alternating, %
then the words $U$ and $V$ appear to be $\varphi$-images, since the length of $Y$ is even. %
However, this contradicts the assumption $h(U), t(U)\neq\lambda$. %
Hence, $Y$ is not letter-alternating. From~Lemma~\ref{UniqueXiEta}, 2) it now follows %
that $h(Y)=c$ and $t(Y)=d$.

Note that $Y\neq cd=c\overline c$, since the word $Y$ is not letter-alternating. %
Therefore,  $|Y|>2$. This completes the proof.
\end{proof}

As a consequence we get the next proposition.

\begin{proposition}
\label{EtaBad}%
If a pair $(U, V)$ is bad, then \[\{h(U), t(U)\}=\{h(V), t(V)\}=\{a,b\}.\]
\end{proposition}

\begin{proof}
Let $\{W_k\}_{k=0}^n$ be an $r$-linking $(U, V)$-series. %
Clearly, if all the pairs $(W_{k-1}, W_k)$ are good, then the pair $(U, V)$ is good as well. %
Hence, there exists an integer $k'\geq1$ such that the pair $(W_{k'{-}1}, W_{k'})$ is bad. %
By Lemma~\ref{EtaNonSave}, we have $\{h(W_{k'}), t(W_{k'})\}=\{a,b\}$. %
The required statement now follows from Proposition~\ref{HTXiEta}. %
\end{proof}

Now consider a bad pair $(YY, YYY)$. %
By Lemma~\ref{EtaNonSave}, we have %
\begin{equation*}
\varphi^{-1}(\eta(YY)) = \varphi^{-1}(\eta(Y)t(Y)h(Y)%
\eta(Y))= \varphi^{-1}(\eta(Y))t(Y)\varphi^{-1}(\eta(Y))  %
\end{equation*}
and
\begin{multline*}
\varphi^{-1}(\eta(YYY)) = \varphi^{-1}(\eta(Y)t(Y)h(Y)%
\eta(Y)t(Y)h(Y)\eta(Y))\\=
\varphi^{-1}(\eta(Y))t(Y)\varphi^{-1}(\eta(Y))t(Y)\varphi^{-1}(\eta(Y)).  %
\end{multline*}
So, we get a pair of the form $(XcX, XcXcX)$, where $X\in\Sigma^+$, $c\in\Sigma$, %
and $XcX\not\sim XcXcX$. %
For the sequel, we need to establish some useful properties of such pairs.

\begin{lemma}
\label{XaXaX} %
Let $XcX\not\sim XcXcX$ for some $X\in\Sigma^+$, $c\in\Sigma$. Then

\textup{1)}\quad $r_1(XcX)= r_1(X)cr_1(X),\quad r_1(XcXcX)=r_1(X)cr_1(X)cr_1(X)$.

\noindent Suppose additionally that the words $XcX$ and $XcXcX$ are $r_1$-reduced. %
Then

\textup{2)}\quad If $X\neq \overline{cc}$, then $XcX, XcXcX\not\in[W]_{r_1}$ %
for any $W\in\mathcal{S}_1$.

\textup{3)}\quad If one of the words $XcX$ and $XcXcX$ is letter-alternating, %
then all the words $XcX$, $XcXcX$, and $X$ are letter-alternating and have odd length.

\textup{4)}\quad If one of the words $XcX$ and $XcXcX$ has a non-uniform tail, %
then the other has the same non-uniform tail. %

\textup{5)}\quad If one of the words $XcX$ and $XcXcX$ has a non-reducible tail, %
then the other has the same non-reducible tail. %

\textup{6)}\quad If the word $XcXcX$ is $\widetilde{A}$-whole ($\widetilde{B}$-whole), %
then both $X$ and $XcX$ are $\widetilde{A}$-whole (resp., $\widetilde{B}$-whole). %
If $X\not\in\{ab,ba\}$ and the word $XcX$ is $\widetilde{A}$-whole ($\widetilde{B}$-whole), %
then  both $X$ and $XcXcX$ are $\widetilde{A}$-whole (resp., $\widetilde{B}$-whole). %

\textup{7)}\quad If the words $XcX$ and $XcXcX$ are $\widetilde{AB}$-whole and have %
no non-reducible tails, then $r(XcX) = r(X)cr(X)$ and $r(XcXcX) = r(X)cr(X)cr(X)$.

\textup{8)}\quad If the words $XcX$ and $XcXcX$ are uniform, %
then the length of $X$ is odd, $h(XcX)=h(XcXcX)=h(X)$, and $t(XcX)=t(XcXcX)=t(X)$.
\end{lemma}

\begin{proof}
Prove the first statement of the lemma. %
Clearly, if \[aaa\in\{X[|X|{-}1]\,X[|X|]a, X[|X|]aX[1],aX[1]\,X[2]\},\] %
then $XaXaX\sim XXX\sim XX\sim XaX$, in contradiction with the lemma's %
condition. Hence, the factors $aaa$ and $bbb$ can occur in the words $XaX$ and $XaXaX$ %
inside the factor $X$ only. This proves the first statement. \medskip

Now suppose that both words $XaX$ and $XaXaX$ are already $r_1$-reduced. Prove the second %
statement. Assume the converse. Then $V\in [W]_{r_1}$  for some words $V\in\{XaX, XaXaX\}$ %
and $W\in\mathcal{S}_1$.  Note that the word $Va$ is either a square or a cube. %
A straightforward check shows that no words $V\not\in[bbabbaabbabb]_{r_1}$
with this property satisfy the regular expressions listed in Lemma~\ref{A1}. %
Hence, $W=bbabbaabbabb$. By Lemma~\ref{A1}, the word $V$ has the following form: %
\[(bba)^2(bba)^*(a(bba)^+)^*(abb)^*(abb)^2.\]
In particular, we get $X[1]=X[|X|]=b$. %
Therefore the factor  $aa$ occurs in $V$ inside the factor  $X$ only. %
It is easy to check that the word $X$  has %
the form $(bba)^2(bba)^*(a(bba)^+)^*(abb)^+$ as a prefix of $V$ and %
has the form $(bba)^+(a(bba)^+)^*(abb)^*(abb)^2$ as a suffix of $V$. %
Combining these forms, we conclude that $X$ has the form above.  %
Obviously, the words $XaX$ and $XaXaX$ has the same form as well. %
Thus we get $X\sim XaX\sim XaXaX$, in contradiction with the lemma's condition. %
So, $V$ is not equivalent to a word from $\mathcal{S}_1$.\medskip

If the word $XaXaX$ ($XaX$) is letter-alternating, then the words $XaX$ and $X$ (resp., $X$) %
are letter-alternating as well. On the other hand, since the factors $aa$ and $bb$ occur in %
$XaXaX$ inside the factor $XaX$ only, the word $XaXaX$ is letter-alternating whenever $XaX$ is. %
As it will be shown in the proof of 8), the length of $X$ is odd whenever the words $XaX$ and %
$XaXaX$ are uniform (in particular, letter-alternating).  Hence, the lengths of both %
words $XaXaX$ and $XaX$ are odd as well. The proof of statement 3 is complete.\medskip %

Prove statements 4--5. Obviously, %
any tail of the word $XaX$ (non-reducible or non-uniform) is the tail of $XaXaX$ as well. %
Conversely, let $T$ be a tail of $XaXaX$.  %
If $|T|\leq|XaX|$, then $T$ is a tail of $XaX$ as well, and there is nothing to prove. %
Suppose that $|T| > |XaX|$. Without loss of generality we may assume that %
$T$ is a left tail. So, $XaX$ is a proper prefix of $T$. %
Consider all possible cases depending on kind of the tail $T$.

Clearly, that $T$ is not a non-uniform $A$-tail. Actually, since the words $XaX$ and $XaXaX$ %
are $r_1$-reduced, the word $X$ cannot begin with $aa$. %
Now suppose that $T$ is a non-uniform $B$-tail, that is, $T=(bba)^{2+k}ab$ for some $k\geq0$. %
Since the word $XaX$ is a proper prefix of $T$,  the word $Xa$ is a prefix %
of the word $(bba)^{2+k}$. Therefore $X$ has the form $(bba)^*bb$. Hence,  the words $XaX$ and %
$XaXaX$ have the form $(bba)^*bb$ as well and have no tails, which is impossible. %

If $T$ is a non-reducible $B$-tail, that is, $T= (bab)^kbaba(ba)^lbb$ for some $k>0$ %
and $l\geq0$, then  $X$ begins with $babb$, in particular, the word $X$ is not letter-alternating. %
On the other hand, the word $T$ does not contain the factor $aa$, therefore $X[|X|]=b$. %
Thus, we have $baba=X[|X|]aX[1..2]\leq XaX$. Since  $XaX$ is %
a prefix of $T$, we conclude that $|Xa|>|(bab)^k|$ and $X\leq (ba)^{l+2}b$.  %
So, the word $X$ appears to be letter-alternating, a contradiction.

Finally, let $T=(aba)^kabab(ab)^laa$ for some $k>0$ and $l\geq0$ ($T$ is a non-reducible $A$-tail). %
Obviously, $X[1]=a$. Since the words $XaX$ and $XaXaX$ are %
$r_1$-reduced, we have $X[|X|]=b$. Also, from $aa=aX[1]\leq aX$ it follows that $|X|< |(aba)^k|$. %
Thus, the word $X$ has the form $(aba)^*ab$. %
It is easy to check that the words $XaX$ and $XaXaX$ have the form $(aba)^*ab$ as well %
and have no tails. This contradiction completes the proof of statements 4--5.
\medskip

Now prove statement 6. First, we prove the following assertion: %
\begin{itemize}%
\item[] \textit{if a word $P$ is a prefix and %
a suffix of a word $Q$ and $Q$ is $\widetilde{A}$- or $\widetilde{B}$-whole, then %
the word $P$ is $\widetilde{A}$- or $\widetilde{B}$-whole respectively as well.} %
\end{itemize}
Indeed, let $Q$ be $\widetilde{A}$-whole and $P=Sa(ab)^kaaT$ for some %
words $S, T\in\Sigma^*$, and integer $k>0$. %
Since the word $P=Sa(ab)^kaaT$ is a prefix of the word $Q$ and $Q$ is $\widetilde{A}$-whole, %
we have $S=S'ab$ for some $S'\in\Sigma^*$. On the other hand, the word $Sa(ab)^kaaT$ is %
a suffix of $Q$, therefore $T=baT'$ for some $T'\in\Sigma^*$. Thus, %
we get $P = S'ab\,a(ab)^kaa\,baT'$, that is, the factor $a(ab)^kaa$ occurs in $P$ inside %
the factor $(aba)(ab^k)a(aba)$. This means that the word $P$ is $\widetilde{A}$-whole, as desired. %
The same argument works for a $\widetilde{B}$-whole word $Q$. So, the first part of %
statement 6 is proved. %

To complete the proof of statement 6, suppose that  the word $XaX$ is $\widetilde{A}$-whole %
(or $\widetilde{B}$-whole) and show that the word $XaXaX$ is $\widetilde{A}$ (resp., or %
$\widetilde{B}$-whole) as well. %
If $XaXaX$ contains no factor of the form $a\widetilde{A}a$ (resp., $b\widetilde{B}b$), %
there is nothing to prove.
Now let $XaXaX=RST$, where $R, T\in\Sigma^*$,  and $S=a(ab)^kaa$ (resp., $S=b(ba)^kbb$)  for %
some integer $k>0$. %
Obviously, either $S\leq XaX$ or $aXa\leq S$.
In the first case, the word $S=a(ab)^kaa$ (resp., $S=b(ba)^kbb$) occurs in $XaXaX$ inside %
the factor $aba(ab)^kaaba$ (resp., $bab(ba)^kbbab$), as desired. %
In the second case, the word $X$ is letter-alternating. %
Clearly, if the length of $X$ is odd, then both words $XaX$ and $XaXaX$ are letter-alternating as well, %
in contradiction with the assumption $S\leq XaXaX$. %
Thus, $X=(ab)^l$ or $X=(ba)^l$ for some integer $l\geq1$ whence %
$XaXaX=(ab)^la(ab)^la(ab)^l$ or $XaXaX=(ba)^la(ba)^la(ba)^l$. %
Since $X\not\in\{ab,ba\}$ by the lemma's condition, we have $l>1$. Obviously, both words $(ab)^la(ab)^l(ab)$ %
and $(ba)^la(ba)^la(ba)^l$  are $\widetilde{AB}$-whole whenever $l>1$. %
This completes the proof of statement 6. \medskip

Now suppose that the words $XaX$ and $XaXaX$ are $\widetilde{AB}$-whole %
and have no non-reducible tails. Prove statement 7. %
By Proposition~\ref{RedSave}, we get $XaX\sim r(XaX)$, $XaXaX\sim r(XaXaX)$. Hence, %
by Proposition~\ref{WholeEqv} and Lemma~\ref{RedTails}, the words $r(X)ar(X)$ and %
$r(X)ar(X)ar(X)$ are $\widetilde{AB}$-whole, have no non-reducible tails, %
and satisfy obvious equalities $r(XaX) = r(r(X)ar(X))$ and $r(XaXaX) = r(r(X)ar(X)ar(X))$. %
So, we will assume that the word $X$ is already $r$-reduced and prove that the %
words $XaX$ and $XaXaX$ are $r$-reduced as well.

First, prove that the word $Xa$ is $r$-reduced. If it is not $r$-reduced, then $Xa=X'a(ab)^kaa$ for %
some $X'\in\Sigma^*$. After the reduction $a(ab)^kaa\to aa$ inside the words $XaX$ and $XaXaX$, %
we obtain two equivalent words: %
\begin{gather*}
XaX = X'\underbrace{a(ab)^kaa}X'a(ab)^ka \to  X'aaX'aa(ba)^k = U,\\
XaXaX =X'\underbrace{a(ab)^kaa}X'\underbrace{a(ab)^kaa}X'a(ab)^ka \to %
X'aaX'aaX'aa(ba)^k = V.  %
\end{gather*}
It follows from Proposition~\ref{RedSave} that $XaX\sim U$ and $XaXaX\sim V$.  %
So, we get $XaX\sim XaXaX$, which is impossible. Hence, the word $Xa$ is $r$-reduced.

In a symmetric way one can prove that the word $aX$ is $r$-reduced as well. %
So, if the word $XaX$ is not $r$-reduced, then it can be decomposed %
as $XaX = PcQaRcS$, where $P, Q, R, S\in\Sigma^*$, $c\in\Sigma$, $PcQ=RcS=X$, and %
the factor $cQaRc$ has the form $a\widetilde{A}a$ or $b\widetilde{B}b$.
Clearly, the prefix $R$ and the suffix $Q$ of the word $X$ do not overlap inside $X$,  %
since the letter-alternating word $R$ cannot contain the factor $cQ[1]=cc$. %
Thus, $X=RTQ$, where either $T=\lambda$ or $T[1] = T[|T|]=c$. %
So, the reduction $cQaRc\to cc$ transforms the pair $(XaX, XaXaX)$ to the pair %
of neighbours  $(U, V)$:
\begin{gather*}
XaX=RT\underbrace{QaR}TQ\to RTTQ = U, \\ %
XaXaX = RT\underbrace{QaR}T\underbrace{QaR}TQ\to RTTTQ = V.
\end{gather*}
Since $XaX\sim U$ and $XaXaX\sim V$, we get %
$XaX\sim XaXaX$, which is impossible. Hence, the word  $XaX$ is $r$-reduced.

Finally, suppose that the word $XaXaX$ is not $r$-reduced %
and let $S$ be a factor of $XaXaX$ of the form $a\widetilde{A}a$ or $b\widetilde{B}b$. %
Then we have $aXa\leq S$, because the word $XaX$ is r-reduced. So, the word $X$ %
is letter-alternating. If the length of $X$ is odd, then the word $XaXaX$ is %
letter-alternating as well, a contradiction with the assumption that the word $XaXaX$ is not $r$-reduced. %
Hence, the length of $X$ is even whence
$XaXaX=(ab)^ka(ab)^ka(ab)^k$ or $XaXaX = (ba)^ka(ba)^ka(ba)^k$ for some $k\geq1$. %
The condition $k=1$ implies that the word $XaXaX$ is not $\widetilde{AB}$-whole. %
From $k>1$ it follows that $r(XaXaX)=XaX = r(XaX)$ and $XaXaX\sim XaX$. %
Since both case are impossible, we finished the proof of statement 7. \medskip

It remains to prove  statement 8. Let $XaX$ and $XaXaX$ be uniform. %
Obviously, the word $X$ is uniform as well. %
First, suppose that the word $X$ is letter-alternating. %
In this case, the length of $X$ is odd, since the words  %
$(ab)^ka(ab)^ka(ab)^k$ and $(ba)^ka(ba)^ka(ba)^k$ are not uniform for any $k\geq1$. %
Hence, all the words $XaX$, $XaXaX$, and $X$ are odd-length letter-alternating %
words, therefore $h(XaX)=h(XaXaX)=h(X)=\lambda$ and $t(XaX)=t(XaXaX)=t(X)=X[|X|]$.%

Now suppose that $X[i..i{+}1]=cc$ for some letter $c$ and integer $i$. %
Then the factor $cc$ occurs in $XaX$ at the $i$-th and at the $|Xa|+i$-th positions. Since %
the word $XaX$ is uniform, we conclude that the length of $X$ is odd. Hence, all the words %
$XaX$, $XaXaX$, and $X$ have odd length. Moreover, by the definition of $\eta$, if $i$ is odd, %
then $h(XaX)=h(XaXaX)=h(X)=X[1]$  and $t(XaX)=t(XaXaX)=t(X)=\lambda$,  %
otherwise  $h(XaX)=h(XaXaX)=h(X)=\lambda$ and $t(XaX)=t(XaXaX)=t(X)=X[|X|]$. %
This completes the proof of the lemma.
\end{proof}

\bigskip
Call especial attention to pairs of inequivalent words $(XcX, XcXcX)$ with $|X|=2$. %
Clearly, the word $XcX$ is overlap-free in any such pair and $[XcX]_{r_1}=XcX$.  %
For the word $XcXcX$, we have %
\[XcXcX\in\{abaabaab, abbabbab, baabaaba, babbabba, bbabbabb, aabaabaa\}=\mathcal{S}_2.\] The next lemma describes
the equivalence classes of all words from $\mathcal{S}_2$.

\begin{lemma}
\label{S2} %
\begin{align*}
[abaabaab]_{r_1}& =(aba)^*(aba)^2ab, & [abbabbab]_{r_1}& =(abb)^*(abb)^2ab,\\
[baabaaba]_{r_1}& =(baa)^*(baa)^2ba, & [babbabba]_{r_1}& =(bab)^*(bab)^2ba, \\
[bbabbabb]_{r_1}& =(bba)^*(bba)^2bb, & [aabaabaa]_{r_1}& =(aab)^*(aab)^2aa.
\end{align*}
\end{lemma}

\begin{proof}
Clearly, any word of the form $(aba)^*(aba)^2ab$ is equivalent to $abaabaab$. %
Conversely, if $W\in[abaabaab]_{r_1}$, then $aW\in[(aab)^3]_{r_1}=[aabaab]_{r_1}$. %
By Lemma~\ref{A1}, we have $aW=(aab)^k$ for some $k\geq2$. %
Hence, $W=(aba)^{k-1}ab$. Since $abaab\not\in[abaabaab]_{r_1}$, we have $k>2$ %
and the word $W$ has the form $(aba)^*(aba)^2ab$, as desired. %
Thus, $[abaabaab]_{r_1}=(aba)^*(aba)^2ab$. %
Now consider the equivalence class $[bbabbabb]_{r_1}$. Obviously, %
this class contains all words of the form $(bba)^*(bba)^2bb$. Conversely, %
let $W\in[bbabbabb]_{r_1}$. Then $Wa\in[bbabba]_{r_1}$ whence $Wa=(bba)^k$, where $k\geq2$, %
by Lemma~\ref{A1}. From $bbabb\not\in[bbabbabb]$ it follows that %
$k>2$ and $W=(bba)^{k-3}(bba)^2bb$. So, we get $[bbabbabb]_{r_1}=(bba)^*(bba)^2bb$. %
The other equivalence classes are symmetric to described above two classes.
\end{proof}

Now we ready to prove the main technical lemma of this section.

\begin{lemma}
\label{Zlo} %
Let $(U, V)$ be a bad pair such that the word $V$ is almost overlap-free. %
Then there exists a word $Y\in\Sigma^+$ such that $V=Y^2$. %
Furthermore, denote the word $Y^3$ by $W$ and let $\{Y_k\}_{k=1}^{\ell(Y)}$, %
$\{V_k\}_{k=1}^{\ell(V)}$, $\{W_k\}_{k=1}^{\ell(W)}$, and $\{U_k\}_{k=1}^{\ell(U)}$ %
be the primary $Y$-, $V$-, $W$-, and $U$-series respectively. Then

\textup{1)}\quad $\eta(U)\sim\eta(W)$.

\textup{2)}\quad $\ell(W)-1\leq\ell(Y)\leq\ell(U)=\ell(W)\leq\ell(V)\leq\ell(W)+1$.

\textup{3)}\quad $U_k\sim W_k\not\sim V_k$ for each $k=2, \dotsc, \ell(W)$. %

\textup{4)}\quad $L_U[k]=L_V[k]=L_W[k]=\lambda$, $R_U[k]=R_V[k]=R_W[k]=\lambda$, %
$h_U[k]=h_V[k]=h_W[k]$, and $t_U[k]=t_V[k]=t_W[k]$ for all $k<\ell(W)$;

$L_Y[k]=R_Y[k]=\lambda$, $h_Y[k]=h_W[k]$, and $t_Y[k]=t_W[k]$ for all $k<\ell(Y)$.

\textup{5)}\quad The words $Y_k$, $V_k$, and $W_k$ are uniform for each $k=1,\dotsc,\ell(W){-}1$. %

\textup{6)}\quad There exists a sequence of letters $\{c_k\}_{k=2}^{\ell(W)}$ %
such that $V_k=Y_kc_kY_k$ and $W_k=Y_kc_kY_kc_kY_k$ for each %
$k=2,\dotsc,\ell(W)$ (if $\ell(Y) = \ell(W){-}1$, then we put $Y_{\ell(W)}=\lambda$). %

\textup{7)}\quad $Y_{\ell(Y)}\in\{ab, ba, aa, bb, a, b \}$.
\end{lemma}

\begin{proof}


Recall that the words $U$ and $V$ are uniform, because the notion of bad pair uses the words $\eta(U)$ and
$\eta(V)$. So, the set of all $r$-linking $(V, U)$-sequences is not empty by Proposition~\ref{Pi_r}. %
For each sequence $\{R_i\}_{i=0}^n$ from this set, %
define \[\beta(\{R_i\}_{i=0}^n)=\Card(\{i\mid (R_{i-1}, R_i)\text{ is bad}, 1\leq i\leq n\})\] %
and \[\gamma(\{R_i\}_{i=0}^n)=\min\{i\mid(R_{i-1}, R_i) \text{ is bad}\}.\] %
Note that the number $\gamma(\{R_i\}_{i=0}^n)$ is well defined: since %
the pair $(V, U)$ is bad, we have  $\beta(\{R_i\}_{i=0}^n)\geq1$.

Among all $r$-linking $(V,U)$-sequences, we choose  the sequence $\{R_i\}_{i=0}^n$ %
having the lexicographically minimal pair $(\beta(\{R_i\}_{i=0}^n), \gamma(\{R_i\}_{i=0}^n))$. %
The proof consists of four steps. In steps 1--3 we %
prove that $\beta(\{R_i\}_{i=0}^n)=1$ and $\gamma(\{R_i\}_{i=0}^n)=1$, that is, %
the sequence $\{R_i\}_{i=0}^n$ has a unique bad pair of neigbours $(V, R_1)$. %
As a by-product of this proof, we will get almost all statements of the lemma. %
Step~4 is a short argument about the pair $(R_1,U)$.\bigskip

\textbf{Step 1.} Let $\overline i = \gamma(\{R_i\}_{i=0}^n)$. We aim to show %
that $\overline i = 1$, i. e., $R_{\overline i-1}=V$.
To simplify the notation we denote %
the words $R_{\overline i-1}$ and $R_{\overline i}$ by $P$ and $Q$ respectively. %
By Lemma~\ref{EtaNonSave}, %
there exists a word $Y$ such that one of the %
words $P$ and $Q$ is equal to $YY$ and the other is equal to $YYY$. %
Consider the primary $Y$-, $P$-, $Q$-, and $V$-series %
$\{Y_k\}_{k=1}^{\ell(Y)}$,  $\{P_k\}_{k=1}^{\ell(P)}$, $\{Q_k\}_{k=1}^{\ell(Q)}$, and  %
$\{V_k\}_{k=1}^{\ell(V)}$ respectively. %
Clearly, $Y_1=Y$, $P_1=P$, $Q_1=Q$, and $V_1=V$.
By the definition of $P$ and $\overline i$, the pair $(V, P)$ is good. %
Hence, we have $\eta(V_1)\sim\eta(P_1)\not\sim\eta(Q_1)$ whence $V_2\sim P_2\not\sim Q_2$. %
Moreover, we obviously have \[L_V[1]=L_P[1]=L_Q[1]=L_Y[1]=R_V[1]=R_P[1]=R_Q[1]=R_Y[1]=\lambda\]
and, by Proposition~\ref{HTXiEta} and Lemma~\ref{EtaNonSave}, %
\[h(V)t(V) = h(P)t(P)=h(Q)t(Q)=h(Y)t(Y)\in\{ab, ba\}.\] %
If we denote $\varphi^{-1}(\eta(Y))$ by $Y'$, then we get %
\[\varphi^{-1}(\eta(YY)) = Y'\varphi^{-1}(t(Y)h(Y))Y'  = Y't(Y)Y'\]
and %
\[\varphi^{-1}(\eta(YYY)) = Y'\varphi^{-1}(t(Y)h(Y))Y'\varphi^{-1}(t(Y)h(Y))Y' %
 = Y't(Y)Y't(Y)Y'
\]
whence $\{P_2, Q_2\}=\{Y_2t(Y)Y_2, Y_2t(Y)Y_2t(Y)Y_2\}$   %
by Lemma~\ref{XaXaX}, 1).\medskip %

Let us prove by induction on $k$ that
\begin{equation}
\label{AncBadPair}
\{P_k, Q_k\} = \{Y_kc_kY_k, Y_kc_kY_kc_kY_k\} \text{ and } %
V_k\sim P_k\not\sim Q_k %
\end{equation}
for some sequence of letters $\{c_k\}_{k=2}^{\ell(Y)}$ and each $k=2, \dotsc, \ell(Y)$. %
In addition, we prove that $|\Anc(Y)|\leq2$.

For $k=2$, we put $c_2=t(Y)$, and  \eqref{AncBadPair} holds, as was shown above. %
Now suppose that \eqref{AncBadPair} holds for some $k$ such that $|Y_k|>2$ %
and prove that $k<\ell(Y)$ and \eqref{AncBadPair} holds for $k+1$ as well.

Obviously, if $|Y_k|>2$, then $|P_k|>2$, $|Q_k|>2$, and $|V_k|>2$ as well. %
By Lemma~\ref{XaXaX}, 2), we have $V_k\not\in\mathcal{S}_1$. From Lemma~\ref{Anc1} and %
Corollary~\ref{TailAOF} it follows %
that $k<\ell(V)$, $r_T(P_k)\sim r_T(V_k)$, and both words $r_T(V_k)$ and $r_T(P_k)$ %
are $\widetilde{AB}$-whole.\smallskip

We claim that $r_T(P_k)=P_k$ and $r_T(V_k)=V_k$, i. e., both words $P_k$ and %
$V_k$ have no non-uniform tails. Assume the converse.  Let $T$ be a tail of the word $P_k$. %
Without loss of generality, $T$ is a left $A$-tail, i. e., $T=(aab)^l(aab)^2ba$ for some $l\geq0$. %
According to Lemma~\ref{XaXaX}, 4), the word $Q_k$ has the tail $T$ as well. %
Since one of the words $P_k$ and $Q_k$ equals $Y_kc_kY_k$, we conclude that %
$T$ is a prefix of $Y_kc_kY_k$. We have $c_k=b$, because $Y_k$ begins with $aa$ and the words $P_k$ and $Q_k$ %
are $r_1$-reduced. If $T\leq Y_k$, then both words $r_T(Y_kc_kY_k)$ and $r_T(Y_kc_kY_kc_kY_k)$ have the factor
$aabaabb$ inside the second occurrence of $Y_k$. So, these words are not $\widetilde{AB}$-whole by definition. But
one of them is $r_T(P_k)$, which is an $\widetilde{AB}$-whole word, a contradiction.

Thus, $Y_k < T$. %
Then $Y_k=(aab)^{l+2}$ and the words $r_T(P_k)$ and $r_T(Q_k)$ %
have the suffix $Y_k$, which is not $\widetilde{AB}$-whole. This contradiction proves %
that $r_T(P_k)=P_k$. Since $P_k\sim V_k$, by Proposition~\ref{TailReg}, 1) we have $r_T(V_k)=V_k$, as desired.
\smallskip

From Lemma~\ref{XaXaX}, 4), 6) it follows that all the words $P_k$, $Q_k$, and $Y_k$  are $\widetilde{AB}$-whole %
and have no non-uniform tails. %
Moreover, since the overlap-free word $V_k$ has no non-reducible tails, the words %
$P_k$, $Q_k$, and $Y_k$ have no non-reducible tails as well.
Hence,  procedure Ancestor cannot stop while processing any of %
the words $V_k$, $P_k$, $Q_k$, and $Y_k$, that is,  %
$k<\min\{\ell(V), \ell(P), \ell(Q), \ell(Y)\}$.\medskip %

By Proposition~\ref{RedSave} and Corollary~\ref{ABUn}, we get %
\[r(V_k) = V_k\sim r(P_k)\not\sim r(Q_k).\] %
At the same time, we have \[\{r(P_k), r(Q_k)\}=\{r(Y_k)c_kr(Y_k), r(Y_k)c_kr(Y_k)c_kr(Y_k)\}\] %
by Lemma~\ref{XaXaX}, 7), 8), where $r(Y_k)$, $r(P_k)$, and $r(Q_k)$ are odd-length uniform words. %
By definition, the function $\eta$ deletes exactly one %
letter from each word $r(Y_k)$, $r(P_k)$, and $r(Q_k)$. %
According to Proposition~\ref{EtaBad}, we get %
\[\eta(V_k)\sim\eta(r(P_k))\not\sim\eta(r(Q_k))\]
whence $V_{k+1}\sim P_{k+1}\not\sim Q_{k+1}.$ %

Additionally, if we put $Y'_k=r(Y_k)$, we get either %
$h(Y'_k)\neq\lambda$ and $t(Y'_k)=\lambda$ or $t(Y'_k)\neq\lambda$ and %
$h(Y'_k)=\lambda$. Consider the first case. %
Then
\begin{equation*}
\begin{split}
&\{\varphi^{-1}(\eta(r(P_k))), \varphi^{-1}(\eta(r(Q_k)))\} = \\
&\{\varphi^{-1}(\eta(Y'_k)c_kh(Y'_k)\eta(Y'_k)), %
\varphi^{-1}(\eta(Y'_k)c_kh(Y'_k)\eta(Y'_k)c_kh(Y'_k)\eta(Y'_k))\}. %
\end{split}
\end{equation*} %
Hence, $h(Y'_k)=\overline c_k$ and we get %
\begin{equation*}
\begin{split}
&\{\varphi^{-1}(\eta(r(P_k))), \varphi^{-1}(\eta(r(Q_k)))\}=\\%
&\{\varphi^{-1}(\eta(Y'_k))\,c_k\varphi^{-1}(\eta(Y'_k)), %
\varphi^{-1}(\eta(Y'_k))\,c_k\varphi^{-1}(\eta(Y'_k))\,c_k\varphi^{-1}(\eta(Y'_k))\}.
\end{split}
\end{equation*} %

In the second case, we have $t(Y'_k)=\overline c_k$ and we get %
\begin{equation*}
\begin{split}
&\{\varphi^{-1}(\eta(r(P_k))), \varphi^{-1}(\eta(r(Q_k)))\}=\\%
&\{\varphi^{-1}(\eta(Y'_k))\,t(Y'_k)\varphi^{-1}(\eta(Y'_k)), %
\varphi^{-1}(\eta(Y'_k))\,t(Y'_k)\varphi^{-1}(\eta(Y'_k))\,t(Y'_k)\varphi^{-1}(\eta(Y'_k))\}.
\end{split}
\end{equation*} %

Thus, if we put $c_{k+1} = \overline{h(Y'_k)}t(Y'_k)$, we get
\[\{P_{k+1}, Q_{k+1}\}=\{Y_{k+1}c_{k+1}Y_{k+1},
Y_{k+1}c_{k+1}Y_{k+1}c_{k+1}Y_{k+1}\} \] %
in all cases in view of Lemma~\ref{XaXaX}, 1). \medskip  %

So, we have proved that
\[\{P_k, Q_k\}=\{Y_kc_kY_k, Y_kc_kY_kc_kY_k\} \text{ and } V_k\sim P_k\not\sim Q_k\] %
for each $k=2, \dotsc, \ell(Y)$, and $|\Anc(Y)|\leq2$. %
Moreover, we have shown that $V_k$ is uniform, the words $V_k$, $P_k$, $Q_k$, and $Y_k$ %
have no tails, $h(V_k)=h(P_k)=h(Q_k)=h(Y_k)$, and $t(V_k)=t(P_k)=t(Q_k)=t(Y_k)$ for all $k <\ell(Y)$. \medskip %

Now consider the words $Y_m=\Anc(Y)$, $Q_m$, $P_m$, and $V_m$, where $m=\ell(Y)$. %
If $Y_m\in\{a, b\}$, then $P_m$ and $Q_m$ are odd-length letter-alternating words. %
Since $V_m\sim P_m$, by Observation~\ref{letteralt} we have %
\[V_m=P_m=Y_mc_mY_m\in\{bab, aba\}\text{ and } %
Q_m = Y_mc_mY_mc_mY_m\in\{babab, ababa\}.\] %
Obviously, we have $\ell(V)=\ell(P)=\ell(Q)=m+1$. %
Note that $t(Y_m)= Y_m$ in this case.  Thus, if we put $Y_{m+1}=\lambda$, %
we get \[V_{m+1}=P_{m+1} = Y_{m+1}c_{m+1}Y_{m+1}\text{ and } %
Q_{m+1}=Y_{m+1}c_{m+1}Y_{m+1}c_{m+1}Y_{m+1},\]%
where $c_{m+1}=\overline{h(Y_m)}t(Y_m)=Y_m$. In the subsequent considerations we put $Y_{m+1}=\lambda$ if
$Y_m\in\{a, b\}$.

If $|Y_m|=2$, then the class $[Y_mc_mY_mc_mY_m]_{r_1}$ contains no %
overlap-free words by Lemma~\ref{S2} and the word $Y_mc_mY_m$ is overlap-free. %
Hence, we have \[V_m=P_m=Y_mc_mY_m\text{ and }Q_m=Y_mc_mY_mc_mY_m.\] %
Clearly, if $Y_m\in\{aa, bb\}$, then the words $V_m$, $P_m$, and $Q_m$ are not %
$\widetilde{AB}$-whole whence $\ell(V)=\ell(P)=\ell(Q)=m$. %
In the case $Y_m\in\{ab, ba\}$, we have $\ell(Q)=m$, since the word $Q$ is not %
$\widetilde{AB}$-whole, and $\ell(V)=\ell(P)=m+1$.

So, in all cases we have $V_m=P_m=Y_mc_mY_m$, $Q_m=Y_mc_mY_mc_mY_m$, and %
\[\ell(Q){-}1\leq\ell(Y)\leq\ell(Q)\leq\ell(P)=\ell(V)\leq\ell(Q){+}1.\]
Hence, $P_k=Y_kc_kY_k$ and $Q_k=Y_kc_kY_kc_kY_k$ for each $k=2, \dotsc, m$,  %
and we get $P=YY$ and $Q=YYY$. %

Notice that the word $V_m=P_m$ is uniform if $\ell(V)=\ell(P)=m+1$.
It is shown above that the words $V_k$ are uniform for all $k<m=\ell(Y)$. %
Hence, the words $V_k$ are uniform for all $k < \ell(V)$.\bigskip

\textbf{Step 2.} Consider the direct normal forms of the words $Y$, $P$, $V$, and $Q$ and %
prove that \[V=N^D(P)=N^D(Y)N^D(Y)\text{ and }N^D(Q)=N^D(Y)N^D(Y)N^D(Y).\] %
Since $\Anc(P)=\Anc(V)$ is an  overlap-free word, the words $P$ and $Q$ share the same main direct normal series. %
Moreover, this series is the primary series of $V$ by Lemma~\ref{Norm4}.

Let $\{\widetilde{Y}_k\}_{k=1}^{\ell(Y)}$  and  $\{\widetilde{Q}_k\}_{k=1}^{\ell(Q)}$ %
be the direct normal series of the words $Y$ and $Q$ respectively. %
Let $m=\ell(Q)$ and put $Y_m=\widetilde{Y}_m=\lambda$ if $m=\ell(Y)+1$. %
Thus, we have  %
\[V_m = Y_mc_mY_m = \widetilde{Y}_mc_m\widetilde{Y}_m \text{ and } %
\widetilde{Q}_m = Q_m = Y_mc_mY_mc_mY_m = \widetilde{Y}_mc_m\widetilde{Y}_mc_m\widetilde{Y}_m.\] %
Now suppose that %
\[V_{k+1}=\widetilde{Y}_{k+1}c_{k+1}\widetilde{Y}_{k+1} \text{ and } %
\widetilde{Q}_{k+1}=\widetilde{Y}_{k+1}c_{k+1}\widetilde{Y}_{k+1}c_{k+1}\widetilde{Y}_{k+1}\] %
for some $k\in[2;m{-}1]$. Prove that the same holds for $k$. From Step 1 we know that either $h_Y[k]=\lambda$ or $t_Y[k]=\lambda$. %
First, assume that  $t_Y[k]=\lambda$. Then $c_{k+1}=c_k=\overline{h_Y[k]}$, %
and we get
\begin{multline*}
V_k= h_V[k]\varphi(V_{k+1})=h_Y[k]\varphi(\widetilde{Y}_{k+1}c_{k+1}\widetilde{Y}_{k+1})=\\%
h_Y[k]\varphi(\widetilde{Y}_{k+1})c_k\overline c_k\varphi(\widetilde{Y}_{k+1}) = %
\widetilde{Y}_kc_k\widetilde{Y}_k %
\end{multline*}
and
\begin{multline*}
\widetilde{Q}_k= h_Q[k]\varphi(\widetilde{Q}_{k+1})=%
h_Y[k]\varphi(\widetilde{Y}_{k+1}c_{k+1}\widetilde{Y}_{k+1}c_{k+1}\widetilde{Y}_{k+1})=\\%
h_Y[k]\varphi(\widetilde{Y}_{k+1})c_k\overline c_k\varphi(\widetilde{Y}_{k+1})%
c_k\overline c_k\varphi(\widetilde{Y}_{k+1}) = %
\widetilde{Y}_kc_k\widetilde{Y}_kc_k\widetilde{Y}_k,
\end{multline*}
as desired.
Now assume that $h_Y[k]=\lambda$. In this case, $c_{k+1} = t_Y[k] = \overline{c_k}$, and we get %
\begin{multline*}
V_k=\varphi(V_{k+1})t_V[k] =  %
\varphi(\widetilde{Y}_{k+1}c_{k+1}\widetilde{Y}_{k+1})t_Y[k] =\\ %
\varphi(\widetilde{Y}_{k+1})\overline c_kc_k\varphi(\widetilde{Y}_{k+1})t_Y[k] = %
\widetilde{Y}_kc_k\widetilde{Y}_k
\end{multline*}
and
\begin{multline*}
\widetilde{Q}_k=\varphi(\widetilde{Q}_{k+1})t_Q[k] =  %
\varphi(\widetilde{Y}_{k+1}c_{k+1}\widetilde{Y}_{k+1}c_{k+1}\widetilde{Y}_{k+1})t_Y[k] = \\%
\varphi(\widetilde{Y}_{k+1})\overline c_kc_k\varphi(\widetilde{Y}_{k+1})%
\overline c_kc_k\varphi(\widetilde{Y}_{k+1})t_Y[k] = %
\widetilde{Y}_kc_k\widetilde{Y}_kc_k\widetilde{Y}_k.
\end{multline*}

So, we have proved that
\[V_k=\widetilde{Y}_kc_k\widetilde{Y}_k \text{ and } %
\widetilde{Q}_k=\widetilde{Y}_kc_k\widetilde{Y}_kc_k\widetilde{Y}_k\] %
for all $k= 2, \dotsc, m$. %
Finally, consider the words $V$ and $\widetilde{Q}_1=N^D(Q)$. %
From Step 1 we have $t_Y[1] = c_2=\overline{h_Y[1]}$.
Thus, we get  %
\begin{multline*}
V = h_V[1]\varphi(V_2)t_V[1] = %
h_Y[1]\varphi(\widetilde{Y}_2c_2\widetilde{Y}_2)t_Y[1]=\\%
h_Y[1]\varphi(\widetilde{Y}_2)t_Y[1]h_Y[1]\varphi(\widetilde{Y}_2)t_Y[1] = %
\widetilde{Y}_1\widetilde{Y}_1 = N^D(Y)N^D(Y)
\end{multline*}
and
\begin{multline*}
\widetilde{Q}_1=h_Q[1]\varphi(\widetilde{Q}_2)t_Q[1]=
h_Y[1]\varphi(\widetilde{Y}_2c_2\widetilde{Y}_2c_2\widetilde{Y}_2)t_Y[1] =\\%
h_Y[1]\varphi(\widetilde{Y}_2)t_Y[1]h_Y[1]\varphi(\widetilde{Y}_2)t_Y[1]h_Y[1]%
\varphi(\widetilde{Y}_2)t_Y[1] =\\%
\widetilde{Y}_1\widetilde{Y}_1\widetilde{Y}_1=N^D(Y)N^D(Y)N^D(Y), %
\end{multline*}%
as desired.

Let $W=N^D(Q)$ and $Y'=N^D(Y)$. By Lemma~\ref{Norm1}, we have $W\sim Q$ and $W_2\sim Q_2$. %
Let $\{S_j\}_{j=0}^l$ be a linking $(W_2,Q_2)$-sequence. Then the sequence %
$\{S'_j=h(Q)\varphi(S_j)t(Q)\}_{j=0}^l$ is an $r$-linking $(W, Q)$-sequence by Lemma~\ref{UniqueXiEta}, 2). %
From Proposition~\ref{HTXiEta} it follows that $h(S'_j)=h(Q)$ and $t(S'_j)=t(Q)$ %
for all $j\leq l$. Hence, $\eta(S'_j)=\varphi(S_j)$ for each $j=0, 1, \dotsc, l$, and %
all pairs $(S'_{j-1}, S'_j)$ for $j\geq1$ are good. %

Recall that $(P, Q)=(R_{\overline{i}-1}, R_{\overline i})$, where
$\overline i=\gamma(\{R_i\}_{i=0}^n)$. Suppose that $\overline i>1$. %
Construct a new $r$-linking $(V, U)$-sequence $\{R'_j\}_{j=0}^{n'}$, where $n'=n+l+1-\overline i$, %
as follows: %
\[R'_0=V, \quad R'_{1+j}=S'_j \text{ for $j=0, \dotsc, l$}, \quad %
R'_{l+1+i-\overline i} = R_i \text{ for $i= \overline i+1, \dotsc, n$}.\] %
Obviously, $\beta(\{R'_j\}_{j=0}^{n'})=\beta(\{R_i\}_{i=0}^n)$ and %
$\gamma(\{R'_j\}_{j=0}^{n'})=1<\gamma(\{R_i\}_{i=0}^n)$, which contradicts  %
the choice of the sequence $\{R_i\}_{i=0}^n$. Hence,
$\gamma(\{R_i\}_{i=0}^n)=1$ whence $P=V$ and $Q = R_1$. %
Since $P=YY$ and $V=Y'Y'$, we conclude that $Y=Y'=N^D(Y)$ and $Q=W$. %
So, we have $V_k=Y_kc_kY_k$ and $W_k=Y_kc_kY_kc_kY_k$ for all $k\leq\ell(Y)$, %
where $\{W_k\}_{k=0}^{\ell(W)}$ is %
the primary $W$-series, and $\widetilde{Y}_k=Y_k$ for all $k\leq\ell(Y)$. %

Moreover, since the word $V_k$ is uniform for each $k=1, \dotsc, \ell(Y){-}1$, the %
word $Y_k$  is uniform for each $k=1, \dotsc, \ell(Y){-}1$ as well. In view of %
Proposition~\ref{RegRed} and Lemma~\ref{XaXaX}, 7), all words $W_k$ are uniform for %
$k<\ell(Y)$ as well. Note that if $\ell(Y)=\ell(W){-}1$, then $\Anc(Y)\in\{a,b\}$. %
So, the words $V_k$ and $W_k$ are uniform for all $k<\ell(W)$ even if $\ell(Y)=\ell(W){-}1$.

This completes the proof of statements 5--7 of the lemma. Additionally, we have proved %
statements 2--4 of the lemma for the words $V$, $W$, and $Y$. It remains to establish %
a connection between the words $W$  and $U$.\bigskip

\textbf{Step 3}.  On this step, we prove that $\beta(\{R_i\}_{i=0}^n)=1$, %
i. e., the pair $(V, W)$ is the only bad pair of neighbours. %
Assume the converse. Let $(P, Q) = (R_{\overline i-1}, R_{\overline i})$, where $\overline i>1$, be the bad pair such that %
all pairs $(R_{i-1}, R_i)$ are good for $1<i<\overline i$. %

By Lemma~\ref{EtaNonSave}, there exists a word $X$ such that $\{P, Q\}=\{XXX, XX\}$. %
Let $\{P_k\}_{k=1}^{\ell(P)}$, $\{Q_k\}_{k=1}^{\ell(Q)}$, and %
$\{X_k\}_{k=1}^{\ell(X)}$ be  the primary $P$-, $Q$-, and $X$-series %
respectively. Obviously, $P_1=P$, $Q_1=Q$, and $X_1=X$. %
One can easily prove (see Step 1) that %
$\{P_2, Q_2\}=\{X_2d_2X_2, X_2d_2X_2d_2X_2\}$, where $d_2=t(X)$, and $W_2\sim P_2\not\sim Q_2$. %
Moreover, we have $h(X)=h(P)=h(Q)=h(W)$ and  $t(X)=t(P)=t(Q)=t(W)$ whence  $d_2=c_2$. %
Now suppose that %
\begin{equation}
\label{AnotherBadPair}
\{P_k, Q_k\}=\{X_kc_kX_k, X_kc_kX_kc_kX_k\}\text{ and }%
W_k\sim P_k\not\sim Q_k %
\end{equation}
for some $k\in[2;\ell(Y){-}1]$, and prove that the same holds for $k+1$.\smallskip  %

First, we show that $|X_k|>2$. Indeed, if $|X_k|=2$, then %
the word $X_kc_kX_k$ is overlap-free, $[X_kc_kX_k]_{r_1}=X_kc_kX_k$,  %
and $X_kc_kX_kc_kX_k\in\mathcal{S}_2$. %
Note that all words from $\mathcal{S}_2$ are  not $\widetilde{AB}$-whole. Since  %
the word $W_k$ is $\widetilde{AB}$-whole, but it is not overlap-free,  %
we get $W_k\not\sim X_kc_kX_k$ and $W_k\not\sim X_kc_kX_kc_kX_k$, contradicting \eqref{AnotherBadPair}. %
In the case $|X_k|=1$, the words $P_k$ and $Q_k$ are letter-alternating. Given $W_k\sim P_k$ and $V_k\leq W_k$, %
we conclude that the words %
$W_k$ and $V_k$ are letter-alternating as well. %
According to Lemma~\ref{XaXaX}, 3), the length of $W_k$ and $V_k$ is odd. The inequality $|Y_k| > 2$ yields %
$|V_k|,|W_k|> 5$ whence $V_k \sim W_k$, contradicting to the condition that $(V,W)$ is a bad pair.\smallskip

So, we have $|X_k|>2$ whence $|P_k|>2$ and $|Q_k|>2$. %
Since $P_k\sim W_k$ and the word $W_k$ is uniform, Lemma~\ref{XaXaX} implies that %
both words $P_k$ and $Q_k$ are $\widetilde{AB}$-whole and contain neither non-uniform nor non-reducible tails. %
Moreover, by Proposition~\ref{RedSave} and Lemma~\ref{XaXaX}, 7) we have %
\[\{r(P_k), r(Q_k)\}=\{r(X_k)c_kr(X_k), r(X_k)c_kr(X_k)c_kr(X_k)\}\] %
and $W_k\sim r(P_k)\not\sim r(Q_k)$. %
Finally, the same argument as in Step 1 gives
\[\{P_{k+1}, Q_{k+1}\}=\{X_{k+1}c_{k+1}X_{k+1}, X_{k+1}c_{k+1}X_{k+1}c_{k+1}X_{k+1}\}\] %
and $W_{k+1}\sim P_{k+1}\not\sim Q_{k+1}$. Also, we get %
\[h(X_k)=h(Q_k)=h(P_k)=h(W_k)  \text{ and } t(X_k)=t(Q_k)=t(P_k)=t(W_k).\]\medskip   %

So, by induction on $k$ we have proved that \[\{P_k, Q_k\}=\{X_kc_kX_k, X_kc_kX_kc_kX_k\}\text{ and } %
W_k\sim P_k\not\sim Q_k\] for each $k=2, \dotsc, \ell(Y).$ %
Let $m=\ell(Y)$. Consider the words $X_m$, $P_m$, and $Q_m$. %
If $Y_m=aa$ or $Y_m=bb$,  then the word $P_m$  has the form %
$(aab)^*(aab)^2aa$ or respectively $(bba)^*(bba)^2bb$ by Lemma~\ref{S2}. %
Clearly, $X_m=(aab)^laa$ (resp., $(bba)^lbb$), where $l\geq0$. %
If $l>0$, then the words $P_m$ and $Q_m$ are equivalent, which is impossible. %
Hence, $l=0$, $X_m=Y_m$, and $\{P_m, Q_m\}=\{W_m, V_m\}$. %

In the case $Y_m=ba$ or $Y_m=ab$, the word $P_m$ has the %
form $(abc_m)^*(abc_m)^2ab$ or $(bac_m)^*(bac_m)^2ba$ respectively. One can easily check %
that $X_m=(abc_m)^lab$ (resp., $(bac_m)^lba$), where $l\geq0$. %
Again, if $l>0$, then $P_m\sim Q_m$, which is impossible. %
Thus, we have $l=0$, $X_m=Y_m$, and $\{P_m, Q_m\} =\{W_m, V_m\}$.

Finally, if $Y_m = a$ or $Y_m=b$, then the word $P_m$ is letter-alternating. %
By Lemma~\ref{XaXaX}, 3), the word $Q_m$ is letter-alternating as well %
and the lengths of the words $X_m$, $P_m$, and $Q_m$ are odd. %
The inequality $|X_m|\geq 2$ yields $|P_m|, |Q_m|\geq 5 $ whence $P_m\sim Q_m$, a contradiction. 
Hence, $|X_m|=1$, and we get $X_m=Y_m$ and $\{P_m, Q_m\}= \{W_m, V_m\}$ again.

So, in all cases we get $\{P_m, Q_m\} = \{W_m, V_m\}$. %
Since  $P_m\sim W_m\not\sim V_m$, we conclude that $P_m=W_m$ and $Q_m=V_m$. %
Hence, we have %
\[\ell(P)=\ell(W), \quad \ell(Q)=\ell(V), \quad \ell(X)=\ell(Y),\]%
and %
\[\Anc(P)=\Anc(W),\quad \Anc(Q)=\Anc(V),\quad \Anc(X)=\Anc(Y).\] %

Obviously, the primary $W$-, $V$-, and $Y$-series are the direct normal series %
of the  words $P$, $Q$, and $X$ respectively. %
In particular, we have $Q_2\sim V_2$ and $Q\sim V$. Let $\{S_j\}_{j=0}^l$ be
a linking $(V_2, Q_2)$-sequence. Then the sequence $\{S'_j=h(V)\varphi(S_j)t(V)\}_{j=0}^l$ %
appears to be an $r$-linking $(V, Q)$-sequence such that %
$\eta(S'_j)=\varphi(S_j)$ for each $j=0, \dotsc, l$.
Hence, all pairs $(S'_{j-1}, S'_j)$ for $j\geq1$ are good. %
Now, if we replace the subsequence $R_0, \dotsc, R_{\overline i}$ from the  $r$-linking $(V, U)$-sequence %
$\{R_i\}_{i=0}^n$ by the sequence $\{S'_j\}_{j=0}^l$, then %
we get a new $r$-linking $(V, U)$-sequence with lesser values of $\beta$ than %
$\beta(\{R_i\}_{i=0}^n)$. We get a contradiction with the choice of $\{R_i\}_{i=0}^n$. %
Hence the pair $(V, W)$ is a unique bad pair among %
$\{(R_{i-1}, R_i)\mid 1\leq i\leq n\}$, as desired.\bigskip

\textbf{Step 4.}
Since all pairs $(R_{i-1}, R_i)$ for $i>1$  are good, we conclude that %
the pair $(W, U)$ is good. %
In addition, the function $\eta$ deletes exactly one letter from the word $W_k$ %
for each $k= 1, \dotsc, \ell(W){-}1$. By  Proposition~\ref{EtaBad}, %
condition~$(**)$ from Lemma~\ref{Anc2} holds for the pair $(W, U)$.  %
The rest of the lemma now follows from Lemma~\ref{Anc2}. %
\end{proof}

\section{Algorithm EqAOF}
\label{AlgEqAOF}
In this section we complete the construction of  Algorithm EqAOF. %
We start with the description of the second main procedure, called Normalize (procedure Ancestor is introduced in
Sect.~\ref{Ancestor}). \medskip

\noindent\textbf{Procedure Normalize.}\\
\textit{Input}. A word $W\sim \mathrm{Anc}(U)$, the arrays $L$, $R$, $h$, $t$ from %
procedure Ancestor$(U)$, and the number $m=\ell(U)$. \\
\textit{Output}. A word $\Norm(U, W)$.\\
\textbf{Step 1.}\ If $m=1$, then return $W$, stop.\\
\textbf{Step 2.}\ Let $m:=m-1$; $W:=h[m]\varphi(W)t[m]$.\\
\textbf{Step 3.}\ If $W = Y^3$ for some $Y\in\Sigma^+$, then set $W:=Y^2$.\\
\textbf{Step 4.}\ Let $W:=L[m]WR[m]$; goto step 1. \\
\textbf{End.}\medskip


Let $\mathcal{S}=\mathcal{S}_1\cup\mathcal{S}_2\cup\{a,b,aa,bb,ab,ba\}$. We ready to construct %
Algorithm EqAOF.\medskip %

\noindent\textbf{Algorithm EqAOF.}\\
\textit{Input}.\ An arbitrary word $U$. \\
\textit{Output}.\ An almost overlap-free word that is equivalent to $U$ or ``FALSE'' if no such almost overlap-free word exists.\\
\textbf{Step 1.}\ Run Ancestor$(U)$ to get $\Anc(U)$, the arrays $L$, $R$, $h$, $t$, and the number $m=\ell(U)$.\\
\textbf{Step 2.}\ Find $W\in\mathcal{S}$ such that $\Anc(U)\sim W$; if no such word $W$ exists, then return ``FALSE'' and stop.\\
\textbf{Step 3.} Run Normalize$(W,L,R,h,t,m)$; Let $V:=\Norm(U, W)$.\\
\textbf{Step 4.}\ If $V$ is almost overlap-free, then return $V$ else return ``FALSE''. \\
\textbf{End.}
\medskip

The next lemma ensures that Algorithm EqAOF works correctly.

\begin{lemma}
\label{Correct}%
For a word $U$, $\EqAOF(U)=V$ if there exists an almost overlap-free word $V\not\in\{aaa, bbb\}$ %
such that $V\sim U$, and $\EqAOF(U)=\mathrm{FALSE}$ otherwise.
\end{lemma}

\begin{proof}
Obviously, if $\EqAOF(U)=V$, then $V$ is an $r_1$-reduced almost overlap-free  word %
and $U\sim V$. Conversely, suppose that  $U\sim V$, where $V$ is an almost overlap-free word %
such that $V\not\in\{aaa, bbb\}$. %
Prove that $\EqAOF(U)=V$. %

Let $\{U_k\}_{k=1}^{\ell(U)}$, $\{V_k\}_{k=1}^{\ell(V)}$ be the primary $U$- and $V$-series %
respectively, and let $m=\min\{\ell(U), \ell(V)\}$. %
If condition~$(**)$ from Lemma~\ref{Anc2} holds, %
then $\ell(U)=\ell(V)$ and $U_k\sim V_k$ for all $k\leq\ell(V)$. %
In particular, we get $\Anc(U)\sim\Anc(V)$. %
By Lemma~\ref{Anc1}, 1), we have $\Anc(V)\in\mathcal{S}$. Thus, Algorithm EqAOF puts $W=\Anc(V)$ %
on Step 2. Since all words $V_k$ for $k\leq\ell(V)$ are almost overlap-free, procedure Normalize %
constructs the main normal $U$-series, which %
coincides with the primary $V$-series by Lemma~\ref{Norm4}. %
So, $\Norm(U, V) = N(U)=V$. Hence, $\EqAOF(U)=V$, as desired.

Now suppose that there exists an integer $k'<m$ such that %
all pairs $(r(r_T(U_{k})), r(r_T(V_{k})))$ for $k< k'$ are good  %
and the pair $(r(r_T(U_{k'})), r(r_T(V_{k'})))$ is bad.
From the proof of Lemma~\ref{Anc2} it follows that  $U_k\sim V_k$ %
for all $k\leq k'$. Moreover, we get $L_U[k]=L_V[k]$, $R_U[k]=R_V[k]$, %
$h_U[k]=h_V[k]$, and $t_U[k]=t_V[k]$ for each $k=1, \dotsc, k'$. 
Let us denote the words $r(r_T(U_{k'}))$ and $r(r_T(V_{k'}))$ by %
$U'$ and $V'$ respectively. By Lemma~\ref{Zlo}, there exists %
a word $Y$ such that $V'=YY$, $\ell(U')=\ell(W')$, and %
$\Anc(U')\sim \Anc(W')\in\mathcal{S}_2$, where $W'=YYY$. %
Let $\{W'_k\}_{k=1}^{\ell(W')}$ and $\{Y_k\}_{k=1}^{\ell(Y)}$ be %
the primary $W'$- and $Y$-series respectively. %
Since $\Anc(U')=\Anc(U)$, Algorithm EqAOF, running on $U$, chooses $W=\Anc(W')$ on Step 2. %

According to Lemma~\ref{Zlo}, the $W$-normal series of $U'$ coincides with %
the primary $W'$-series. Since $|W'_m| = |Y_mc_mY_mc_mY_m|=3|Y_m|+2$, the word $W'_m$ %
is not a cube for each $m = 2, \dotsc, \ell(W')$. Hence, the condition on Step 3 %
of procedure Normalize is not fulfilled, and procedure Normalize, running on $W$, constructs %
the $W$-normal series of $U'$, while $m=\ell(U), \dotsc, \ell(U){-}\ell(U'){+}2$. %

On the iteration with $m=\ell(U){-}\ell(U'){+1}=k'$ we obtain the word $W'=YYY$. %
On Step 3 of this iteration, procedure Normalize reduces $W'$ to the word $V'=YY$. %
Since the word $V_{k'}$ is almost overlap-free, we get $V'=r_T(V_{k'})$. %
Hence, on Step 4  we obtain the word $L_U[k']V'R_U[k']=L_V[k']r_T(V_{k'})R_V[k'] = V_{k'}$. 
After that procedure Normalize restores all words $V_m$ for $m=k'{-}1, \dotsc, 1$ and %
returns the word $\Norm(U, W)=V$. So, Algorithm EqAOF returns $V$, as desired.
\end{proof}

The next lemma  estimates the time complexity of Algorithm EqAOF. %

\begin{lemma}
\label{Time} %
Algorithm EqAOF has $O(|U|)$ time complexity, where $U$ is an input word.
\end{lemma}

\begin{proof}
Step 1 runs procedure Ancestor. The cycle in Ancestor consists of %
constant-time and linear-time operations (checks and reductions). The %
nontrivial check in Step~1 of Ancestor is linear, because two given %
classes are recognizable languages. So, if the cycle in Ancestor %
is bounded by $C|U|+D$ for any given word $U$, then %
the complexity of Ancestor is bounded by $\sum_{k=1}^{\ell(U)}(C|U_k|+D)$. %
Since $|U_k|\leq 2^{-k+1}|U|$ and $\ell(U)\leq\log |U|$,  we %
conclude that the time complexity of Ancestor is bounded by %
\[C|U|\sum_{k=0}^{\infty}\frac{1}{2^k} + D\log |U| = 2C|U| + D\log |U|.\] %
Hence, procedure Ancestor runs in $O(|U|)$ time.

As to the complexity of Step 2 of Algorithm EqAOF, from Lemmas~\ref{A1} and~\ref{S2} %
it follows that the equivalence class of any word of $\mathcal{S}$ is a recognizable language. %
Thus, the word $\mathrm{Anc}(U)$ is examined by a finite set of fixed %
finite automata, and the complexity of this step is linear with respect %
to $|\mathrm{Anc}(U)|$. Procedure Normalize, applied on Step~3, runs %
in $O(|U|)$ time by the same reason as procedure Ancestor. Finally, a word can be %
checked for almost overlap-freeness in linear time also. %
The corresponding algorithm can be constructed, for
example, by modifying Algorithm A$'$ from \cite{Shur}. We see that %
Algorithm~EqAOF has $O(|U|)$ time complexity.
\end{proof}

Lemmas~\ref{Correct} and~\ref{Time} together prove Theorem~\ref{Th2}.


\end{document}